\documentclass[10pt,a4paper]{article}

\usepackage{amsmath,amsgen,latexsym}
\usepackage{amstext,amssymb,amsfonts,latexsym}
\usepackage{theorem}
\usepackage{pifont}

 \newcommand{\bs}{\bigskip}
 \newcommand{\ms}{\medskip}
 \newcommand{\n}{\noindent}
 \newcommand{\s}{\smallskip}
 \newcommand{\hs}[1]{\hspace*{ #1 mm}}
 \newcommand{\vs}[1]{\vspace*{ #1 mm}}

 \newcommand{\setempty}{\mathrm{\O}}
 \newcommand{\real}{\mathbb{R}}

 \newcommand{\nat}{\mathbb{N}}
 
 \newcommand{\integer}{\mathbb{Z}}
 \newcommand{\rational}{\mathbb{Q}}
 
 \newcommand{\complex}{\mathbb{C}}

 \newcommand{\algebraic}{\mathbb{A}}

 \newcommand{\ie}{\textrm{i.e.},\hspace*{2mm}}
 \newcommand{\eg}{\textrm{e.g.},\hspace*{2mm}}
 \newcommand{\etal}{\textrm{et al.}\hspace*{2mm}}
 \newcommand{\etalc}{\textrm{et al.}}

 \newcommand{\AAA}{{\cal A}}

 \newcommand{\FF}{{\cal F}}
 
 \newcommand{\HH}{{\cal H}}
 
 \newcommand{\GG}{{\cal G}}
 \newcommand{\PP}{{\cal P}}
 \newcommand{\UU}{{\cal U}}

 \newcommand{\p}{\mathrm{P}}
 \newcommand{\np}{\mathrm{NP}}

 \newcommand{\fp}{\mathrm{FP}}
 \newcommand{\sharpp}{\#\mathrm{P}}

\theoremstyle{plain}
\theoremheaderfont{\bfseries}
 \newtheorem{theorem}{Theorem}[section]
 \newtheorem{lemma}[theorem]{Lemma}
 \newtheorem{proposition}[theorem]{Proposition}
 \newtheorem{corollary}[theorem]{Corollary}

 {\theorembodyfont{\rmfamily}
  }
 {\theorembodyfont{\rmfamily} }
 {\theorembodyfont{\rmfamily} }

 \newtheorem{claim}{Claim}

 \newenvironment{proof}{\par \noindent
            {\bf Proof. \hs{2}}}{\hfill$\Box$ \vspace*{3mm}}

 \newenvironment{proofsketch}{\par \noindent
            {\bf Proof Sketch. \hs{2}}}{\hfill$\Box$ \vspace*{3mm}}

 \newenvironment{proofof}[1]{\vspace*{5mm} \par \noindent
         {\bf Proof of #1.\hs{2}}}{\hfill$\Box$ \vspace*{3mm}}

 \newcommand{\pair}[1]{\langle #1 \rangle}
 
\newcommand{\ignore}[1]{}

\newcommand{\holant}{\mathrm{Holant}}

\newcommand{\sharpcsp}{\#\mathrm{CSP}}
\newcommand{\sharpcspstar}{\#\mathrm{CSP}^{*}}
\newcommand{\sharpcspplus}{\#\mathrm{CSP}^{+}}

\newcommand{\APreduces}{\leq_{\mathrm{AP}}}
\newcommand{\APequiv}{\equiv_{\mathrm{AP}}}
\newcommand{\DG}{{\cal DG}}
\newcommand{\NZ}{{\cal NZ}}
\newcommand{\ED}{{\cal ED}}
\newcommand{\IM}{{\cal IM}}
\newcommand{\csp}{\mathrm{csp}}

\begin{document}
\pagestyle{plain}
\setcounter{page}{1}

\begin{center}
{\Large {\bf Approximate Counting for Complex-Weighted Boolean \s\\
Constraint Satisfaction Problems}}\footnote{This paper improves an initial report that appeared in the Proceedings of the 8th Workshop on Approximation and Online Algorithms (WAOA 2010), Liverpool, United Kingdom, September 9--10, 2010, Lecture Notes in Computer Science, Springer-Verlag, Vol.6534, pp.261--272, 2011.} \bs\ms\\
{\sc Tomoyuki Yamakami}\footnote{Current Affiliation: Department of Information Science, University of Fukui, 3-9-1 Bunkyo, Fukui 910-8507, Japan} 
\end{center}

\bs

\begin{quote}
\n{\bf Abstract:} 
Constraint satisfaction problems (or CSPs) have been extensively 
studied in, for instance, artificial intelligence, database theory, 
graph theory, and statistical physics. {}From a practical  
viewpoint, it is beneficial  to approximately solve those CSPs. 
When one tries to approximate the 
total number of truth assignments that satisfy all Boolean-valued  
constraints for (unweighted) Boolean CSPs, there  is a known 
trichotomy theorem by which all such counting problems are neatly classified into exactly three categories under polynomial-time (randomized) approximation-preserving reductions.  
In contrast, we obtain a dichotomy theorem of approximate counting for 
complex-weighted Boolean CSPs, provided that all complex-valued unary 
constraints are freely available to use. It is the expressive power of free unary constraints that enables us to prove such a stronger, complete classification theorem. 
This discovery makes a step forward in the quest for the approximation-complexity  
classification of all counting CSPs.
To deal with complex weights, we employ proof techniques of factorization and arity reduction  
along the line of solving Holant problems.  
Moreover, we introduce a novel notion 
of T-constructibility that naturally induces approximation-preserving reducibility. 
Our result also gives an approximation analogue of the dichotomy 
theorem on the complexity of 
exact counting for complex-weighted Boolean CSPs.

\n{\bf Keywords:} 
constraint satisfaction problem, 
constraint,  
T-constructibility,  
Holant problem,  
signature,  
approximation-preserving reduction,  
dichotomy theorem

\n{\bf AMS subject classification:} 68Q15, 68Q17, 68W20, 68W25, 68W40
\end{quote}

\section{Background, New Challenges, and Achievement}\label{sec:introduction}

{\em Constraint satisfaction problems} (or CSPs) have appeared in many different contexts, such as graph theory, database theory, type inferences, scheduling, and notably artificial intelligence, 
from which the notion of CSPs was originated. The importance of CSPs comes partly from the fact that the framework of the CSP is broad enough to capture numerous natural problems arising in real applications. 
Generally, an input instance of a CSP is a set of ``variables'' 
(over a specified domain) and a set of ``constraints'' (such a set of constraints is 
sometimes called a {\em constraint language}) among these variables. 
We are particularly interested in the case of {\em Boolean variables} throughout this paper. 

As a decision problem, a CSP asks whether there exists 
an appropriate variable assignment that satisfies all the given constraints. In particular, Boolean-valued   constraints (or simply, {\em Boolean constraints}) can be expressed 
by Boolean functions or equivalently propositional logical formulas;  
hence, the CSPs with Boolean constraints are all $\np$ problems. 
Typical examples of CSP are the satisfiability problem    
(or SAT), the vertex cover problem, and the colorability problem, 
all of which are known to be $\np$-complete. On the contrary,  other CSPs, 
such as the Horn satisfiability (or HORNSAT), fall into $\p$. 
One naturally asks 
what kind of constraints make them $\np$-complete or solvable efficiently. To be more precise, 
we first restrict our attention on CSP instances that depend only on constraints chosen from a given set $\FF$ of constraints. 
Such a restricted CSP is conventionally denoted $\mathrm{CSP}(\FF)$. 
A classic {\em dichotomy theorem} of Schaefer \cite{Sch78} states that 
if $\FF$ is included in one of six clearly specified classes,\footnote{These classes are defined in terms of $0$-valid, $1$-valid, weakly positive, weakly negative, affine, and bijunctive constraints. See \cite{Sch78} for their definitions.} $\mathrm{CSP}(\FF)$ belongs to $\p$; otherwise, it is indeed $\np$-complete. To see the significance of this theorem, let us  recall a result of Ladner \cite{Lad75}, who demonstrated under the   $\p\neq\np$ assumption that 
all $\np$ problems fill 
infinitely many disjoint categories located between the class $\p$ and the class of $\np$-complete problems. Schaefer's claim implies that 
there are no intermediate categories for Boolean CSPs. 

Another challenging question is to count the number of satisfying assignments for a given CSP instance. The counting satisfiability problem, $\#\mathrm{SAT}$, is a typical counting CSP (or succinctly, $\sharpcsp$), 
which is known to be complete for Valiant's counting class $\sharpp$ \cite{Val79a}. Restricted to a set $\FF$ of Boolean constraints,   
Creignou and Hermann \cite{CH96} gave a dichotomy theorem,  
concerning the computational complexity of the restricted counting problem 
$\sharpcsp(\FF)$. 
\begin{quote}
If all constraints in $\FF$ are affine,\footnote{An affine relation is a set of solutions of a certain set of linear equations over $\mathrm{GF}(2)$.} then $\sharpcsp(\FF)$ is solvable in polynomial time. Otherwise, $\sharpcsp(\FF)$ is $\sharpp$-complete under polynomial-time Turing reductions.
\end{quote}
In real applications, constraints often take real numbers, and this fact leads us to concentrate on ``weighted'' \#CSPs (namely, \#CSPs with arbitrary-valued constraints). In this direction, 
Dyer, Goldberg, and Jerrum \cite{DGJ09} extended the above result to  
nonnegative-weighted Boolean \#CSPs. Eventually, Cai, Lu, and Xia 
\cite{CLX09x,CLX09}  pushed the scope of Boolean \#CSPs further to 
complex-weighted Boolean \#CSPs, and thus all Boolean \#CSPs have been completely classified in the following fashion.
\begin{quote}
If either all complex-valued constraints in $\FF$ are in a set $\AAA$ of ``generalized'' affine or they are in a set $\PP$ of ``product-type'' constraints,\footnote{More precisely, this set, which was originally introduced as $\PP$ by Cai \etalc~\cite{CLX09x,CLX09}, is composed of products of the  equality/disequality constraints (which 
will be explained in Section \ref{sec:constraint-set}) together with unary constraints. Our initial report \cite{Yam10} used an alternative notation of ``$\ED$,'' emphasizing the importance of its key components, {\em equality}  and {\em disequality}. To be consistent with the report, we will contnue using  ``$\ED$'' throughout this paper. } then $\sharpcsp(\FF)$ is solved in polynomial time. Otherwise, $\sharpcsp(\FF)$ is $\sharpp$-hard under polynomial-time Turing reductions.  
\end{quote}

When we turn our attention from exact counting 
to {\em (randomized) approximate counting}, however, a situation seems much more complicated and its landscape looks quite different. 
Instead of the aforementioned dichotomy theorems on the exact-counting model, Dyer, Goldberg, and Jerrum \cite{DGJ10} presented 
a {\em trichotomy theorem} regarding the complexity of approximately counting the number of satisfying assignments for each Boolean CSP instance. 
What they actually proved 
is that, depending on the choice of a set $\FF$ of Boolean constraints, the complexity of approximately solving $\sharpcsp(\FF)$ can be classified into  exactly three categories.   
\begin{quote}
If all constraints in $\FF$ are affine, then $\sharpcsp(\FF)$ is polynomial-time solvable. Otherwise, if all constraints in $\FF$ belong to 
a well-defined class, known as $IM_2$, 
then $\sharpcsp(\FF)$ is equivalent in complexity 
to $\#\mathrm{BIS}$. Otherwise, 
$\sharpcsp(\FF)$ is equivalent to $\#\mathrm{SAT}$. 
The equivalence is defined  
via {\em polynomial-time (randomized) approximation-preserving (Turing) reductions} 
(or AP-reductions, in short).
\end{quote}
Here, \#BIS is the problem of counting the number of independent sets in a given bipartite graph. 

There still remains a nagging question on the approximation complexity 
of a ``weighted'' version of Boolean \#CSPs: what happens if we expand the scope of Boolean \#CSPs from unweighted problems to real/complex-weighted ones? Unfortunately, there have been few results showing the hardness of approximately solving \#CSPs with real/complex-valued constraints, 
except for, \eg \cite{GJ07}. 
Unlike Boolean constraints, when we deal with complex-valued constraints, a significant complication occurs as a result of massive cancellations of weights in the process of summing 
all weights produced by given constraints. This situation demands a distinctive approach toward an analysis of the complex-weighted \#CSPs. Do we still have a classification theorem similar to the theorem of 
Dyer \etal or something quite different? 
In this paper, we will answer this question under a reasonable assumption 
that all unary (\ie arity $1$) constraints are freely available to use. 
Meanwhile, let the notation 
$\sharpcspstar(\FF)$ denote the complex-weighted counting problem $\sharpcsp(\FF\cup\UU)$, in which $\UU$ expresses the set of all complex-valued unary constraints.  
A free use of such unary constraints appeared in the past literature for, \eg CSPs \cite{DF03}, degree-bounded \#CSPs \cite{DGJR09}, and Holant problems \cite{CLX09x,CLX@}. 
Although it is reasonable, this extra assumption draws a clear distinction between the approximation 
complexity of $\sharpcspstar(\FF)$ and that of 
$\sharpcsp(\FF)$, except for the case of Boolean  constraints. 
If we restrict our interest on Boolean constraints, then the only nontrivial unary constraints are $\Delta_{0}$ and $\Delta_{1}$ (which are called ``constant constraints'' and will be explained in Section \ref{sec:constraint}) and thus,  
as shown in \cite{DGJ10}, 
we can completely eliminate them from the definition of $\sharpcspstar(\FF)$ using polynomial-time randomized approximation algorithms. The elimination of those constant constraints is also possible in our general setting of complex-weighted \#CSPs when the values of all constraints are limited to algebraic complex numbers.  

Regarding the approximation complexity of $\sharpcspstar(\FF)$'s, 
the expressive power of unary complex-valued constraints leads us to 
a surprising {\em dichotomy theorem}---Theorem \ref{dichotomy-theorem}---which depicts a picture that looks markedly 
different from that of Dyer \etal     
\begin{theorem}\label{dichotomy-theorem}
Let $\FF$ be any set of complex-valued constraints. If  
$\FF$ is included in $\ED$ (expressed as $\PP$ in \cite{CLX09x,CLX09}), then $\sharpcspstar(\FF)$ is solvable in polynomial time. 
Otherwise, $\sharpcspstar(\FF)$ is $\#\mathrm{SAT}_{\complex}$-hard (\ie at least as hard as $\#\mathrm{SAT}_{\complex}$) under AP-reductions. 
\end{theorem}
Here, the counting problem $\#\mathrm{SAT}_{\complex}$ is a complex-weighted analogue of $\#\mathrm{SAT}$.  

Theorem \ref{dichotomy-theorem} marks a significant progress in quest of  determining the approximation complexity of all counting problems $\sharpcsp(\FF)$ in the most general form. 
This theorem also bears a close resemblance to the aforementioned dichotomy theorem of Cai \etal on the  exact-counting model when all unary constraints are assumed to be freely available. To be more precise, we can deduce from the theorem of Cai \etal the following corollary: for any constraint set $\FF$, if $\FF$ is included in $\ED$ (denoted $\PP$ in \cite{CLX09x,CLX09}) then $\sharpcspstar(\FF)$ belongs to $\fp_{\complex}$ and, otherwise, $\sharpcspstar(\FF)$ is $\sharpp$-hard under polynomial-time Turing reductions. This comes from a fact that, since the set $\AAA$ does not contain $\UU$,  $\sharpcspstar(\AAA\cup \UU)$ (which equals $\sharpcspstar(\FF)$) must be $\sharpp$-hard as a consequence of Cai \etalc's dichotomy theorem. This resemblance is in fact induced from the powerful expressibility of free unary constraints.  

Our proof of Theorem \ref{dichotomy-theorem} heavily relies on the previous work of Dyer 
\etalc~\cite{DGJ09,DGJ10} and, particularly, the work of 
Cai \etalc~\cite{CLX09x,CLX09}, 
which is based on a theory of {\em signature} 
(see \cite{CL08,CL11}) that formulate underlying concepts of {\em holographic algorithms} (which are Valiant's \cite{Val02a,Val02b,Val06,Val08} manifestation 
of a new algorithmic design method of solving seemingly-intractable counting problems in polynomial time). 
A challenging issue for this paper is that core arguments of Dyer \etalc~\cite{DGJ10} 
exploited Boolean natures of Boolean  constraints and 
they are not designed to lead to a dichotomy theorem for 
complex-valued constraints. 
Cai's theory of signature, on the contrary,  deals with complex-valued  
constraints (which are formally called {\em signatures}); however, the theory has been developed 
over polynomial-time Turing reductions but it is not meant to be valid under AP-reductions. For instance, a useful technical tool known as {\em polynomial interpolation}, which is frequently used in an analysis of exact-counting of Holant problems, is no longer applicable in general.  
Therefore, our first task is to re-examine the well-known 
results in this theory and salvage its key arguments that are still valid  for our AP-reductions. 
{}From that point on, we need to find our own way to establish an approximation theory. 

Toward forming a solid approximation theory, a notable technical tool developed in this paper 
is a notion of {\em T-constructibility} in tandem with the aforementioned AP-reducibility. Earlier, Dyer \etalc~\cite{DGJ09} utilized an existing notion of (faithful and perfect) {\em implementation} for Boolean   constraints in order to induce their desired AP-reductions. The T-constructibility similarly maintains the validity of the AP-reducibility; in addition, it is more suitable to handle complex-valued constraints in a more systematic fashion.    
Other proof techniques involved in proving our main theorem include   
(1) {\em factorization} (of Boolean parts) of complex-valued constraints and (2) {\em arity reduction} of complex-valued constraints. Factoring complex-valued constraints helps us conduct crucial analyses on fundamental properties of those constraints, and reducing the arities of constraints helps construct, from given constraints of higher arity, binary constraints, which we can handle directly by a case-by-case analysis. In addition, a particular binary constraint---$Implies$---plays a pivotal role in the proof of Theorem \ref{dichotomy-theorem}.  This situation is clearly distinguished from \cite{CLX09x,CLX@,CLX09}, which instead utilized the affine property. 


To prove our dichotomy theorem, we will organize the subsequent sections in the following fashion. Section \ref{sec:basic-definition} will give the detailed descriptions of our key terminology: constraints, Holant problems, counting CSPs, and 
AP-reductions. In particular, an extension of the existing notion of 
randomized approximation scheme over non-negative integers to 
arbitrary complex numbers will be described in Section \ref{sec:randomized-scheme}. Briefly explained in 
Section \ref{sec:constructibility} is the concept of T-constructibility, 
a technical tool developed exclusively in this paper. 
For readability, a basic property of T-constructibility will be proven in Section \ref{sec:elimination}. 
Section \ref{sec:constraint-set} will introduce several crucial sets of constraints, which are bases of our key results.  
Toward our main theorem, we will develop solid 
foundations in Sections \ref{sec:elementary-reduction} 
and \ref{sec:support-reduction}. Notably, a free use of ``arbitrary'' unary constraint is heavily required in Section \ref{sec:elementary-reduction} to prove several approximation-complexity bounds of $\sharpcspstar(f)$.  
As an important ingredient of the proof of the 
dichotomy theorem, we will present in Section \ref{sec:IM2-support}   approximation  hardness of  $\sharpcspstar(f)$ for two types of constraints $f$. 
The dichotomy theorem will be finally proven in Section \ref{sec:main-theorem}, achieving the goal of this paper.   

Given a constraint, if its outcomes are limited to algebraic complex numbers, we succinctly call the constraint an {\em algebraic constraint}. 
When all input instances are only algebraic constraints, as we noted earlier, we can further eliminate the constant constraints and thus strengthen the main theorem. To describe our next result, we introduce a special notation $\sharpcspplus_{\algebraic}(\FF)$ to indicate $\sharpcspstar(\FF)$ in which (i) all input instances are limited to algebraic constraints and (ii) free unary constraints take neither of the forms $c\cdot \Delta_0$ nor $c\cdot\Delta_1$ for any constant $c$.  
Similarly, $\#\mathrm{SAT}_{\algebraic}$ is induced from $\#\mathrm{SAT}_{\complex}$ by limiting node-weights within algebraic complex numbers.  The power of AP-reducibility helps us establish the following corollary of the main theorem.
\begin{corollary}\label{algebraic-main-theorem}
Let $\FF$ be any set of complex-valued constraints. If $\FF\subseteq\ED$, then $\sharpcspplus_{\algebraic}(\FF)$ is solved in polynomial time; otherwise, $\sharpcspplus_{\algebraic}(\FF)$ is 
$\#\mathrm{SAT}_{\algebraic}$-hard under AP-reductions.  
\end{corollary}

This corollary will be proven in Section \ref{sec:main-theorem}. A key to the proof of the corollary is an AP-equivalence between $\sharpcspstar_{\algebraic}(\FF)$ and $\sharpcspplus_{\algebraic}(\FF)$ 
for any constraint set $\FF$, where the subscript ``$\algebraic$'' in $\sharpcspstar_{\algebraic}(\FF)$ emphasizes the restriction on input instances within algebraic constraints. This AP-equivalence is a direct consequence of the elimination of  $\Delta_0$ and $\Delta_1$ from $\sharpcspstar_{\algebraic}(\FF)$ and this elimination will be demonstrated in Section \ref{sec:elimination}.


\ms
\n{\bf Outline of the Proof of the Main Theorem:}\hs{1} 
The proof of our dichotomy theorem (Theorem \ref{dichotomy-theorem}) is outlined as follows. 
First, we will establish in Section \ref{sec:upper-bound} the equivalence between $\#\mathrm{SAT}_{\complex}$ and $\sharpcspstar(OR)$, where $OR$ represents the logical ``or'' on two Boolean variables.   
This makes it possible to work solely with $\sharpcspstar(OR)$, instead of 
$\#\mathrm{SAT}_{\complex}$ in the subsequent sections.  
When a constraint set $\FF$ is completely included in $\ED$, we will show  in Lemma \ref{basic-case-FPC} that $\sharpcspstar(\FF)$ is polynomial-time solvable. On the contrary, when $\FF$ is not included in $\ED$, 
we choose a constraint $f$ not in $\ED$. Such a constraint will be treated by Proposition \ref{key-proposition}, in which we will AP-reduce $\sharpcspstar(OR)$ to $\sharpcspstar(f)$. The proof of this proposition will be split into two cases, depending on whether or not $f$ has ``imp support,'' which is a property associated with the constraint $Implies$. When $f$ has such a property, Proposition \ref{no-affine-and-IM2} helps demonstrate the hardness of $\sharpcspstar(f)$, namely, an AP-reduction of $\sharpcspstar(f)$ from $\sharpcspstar(OR)$ and thus from $\#\mathrm{SAT}_{\complex}$. In contrast, if $f$ lacks the property, then we will examine two subcases. If $f$ is a non-zero constraint, then Lemma \ref{affine-PP-induction} together with Proposition \ref{1-x-y-z-case} will 
lead to the hardness of 
$\sharpcspstar(f)$. Otherwise, Proposition \ref{IM-XOR-IMP} will establish   the desired AP-reduction. Therefore, the proof of the theorem is completed.

\ms

Now, we begin with an explanation of basic definitions.

\section{Basic Definitions}\label{sec:basic-definition}

This section briefly presents fundamental notions and notations that 
will be used in later sections. For any finite set $A$, the notation $|A|$ 
denotes the {\em cardinality} of $A$. 
A {\em string} over an alphabet $\Sigma$ is a finite sequence of symbols from $\Sigma$ and $|x|$ denotes the {\em length} of a string $x$, where an {\em alphabet} is a non-empty finite set of ``symbols.''   
Let $\nat$ denote the set of all {\em natural numbers} (\ie non-negative integers). For convenience, $\nat^{+}$ denotes $\nat-\{0\}$. 
 For each number 
$n\in\nat$, $[n]$ expresses the integer set $\{1,2,\ldots,n\}$. 
Moreover, $\real$ and $\complex$ denote respectively the sets of all real numbers and of all complex numbers.  Given a complex number $\alpha$,  
let $|\alpha|$ and $\arg(\alpha)$ respectively denote the {\em absolute value} and the {\em argument} of $\alpha$, where we always assume that $-\pi<\arg(\alpha)\leq\pi$. The special notation $\algebraic$ represents the set of all {\em algebraic complex numbers}. 
The notation $A^{T}$ for any matrix $A$ indicates the {\em transposed matrix} of $A$.  We always treat ``vectors'' as {\em row vectors}, unless stated otherwise.  

For any undirected graph $G=(V,E)$ (where $V$ is a {\em node set} and $E$ is an {\em edge set})  and a node $v\in V$, an {\em incident set} $E(v)$ of $v$ is the set of all edges incident on $v$, and $deg(v)=|E(v)|$ 
is the {\em degree} of $v$. 
When we refer to labeled nodes in a given undirected graph, unless there is any ambiguity, we call such nodes by their labels instead of their original node names. For example, if a node $v$ has a label of Boolean variable $x$, then we often call it ``node $x$,'' although there are many other nodes labeled $x$, as far as it is clear from the context which node is referred to. Moreover, when $x$ is a Boolean variable, as in this example, we succinctly call any node labeled $x$ a ``variable node.'' 

\subsection{Constraints, Signatures, Holant Problems, and \#CSP}\label{sec:constraint}

The most fundamental concept in this paper is ``constraint'' on the 
Boolean domain. 
A function $f$ is called a {\em (complex-valued) constraint} of arity $k$ 
if it is a function from $\{0,1\}^k$ to $\complex$.  Assuming the standard lexicographic order on $\{0,1\}^{k}$, we express $f$ as a series of its output values, which is identified with an element in the complex space $\complex^{2^{k}}$. 
For instance, if $k=1$, then $f$ equals $(f(0),f(1))$, and if $k=2$, 
then $f$ is expressed as  $(f(00),f(01),f(10),f(11))$. 
 A constraint $f$ is {\em symmetric} if the values of $f$ depend only on the {\em Hamming weights} of inputs; otherwise, $f$ is called 
{\em asymmetric}.  When $f$ is a 
symmetric constraint of arity $k$, we use another notation  $f=[f_0,f_1,\ldots,f_k]$, where each $f_i$ is the value of $f$ on inputs of 
Hamming weight $i$. As a concrete example, when $f$ is the equality function ($EQ_k$) of arity $k$, it is expressed as $[1,0,\ldots,0,1]$ (including $k-1$ zeros). Let us recall from Section \ref{sec:introduction} the set $\UU$ of all unary constraints 
and we use the following special unary constraints: $\Delta_0 =[1,0]$ and $\Delta_1=[0,1]$. These constraints are often referred to as ``constant constraints.'' 

Before introducing \#CSPs, we will give a brief description of Holant problem; however, we focus our attention only on ``bipartite Holant problems''  whose input instances are ``signature grids'' containing bipartite graphs $G$, in which all nodes on the left-hand side of $G$ are labeled by signatures in $\FF_1$ and all nodes on the right-hand side of $G$ are labeled by signatures in $\FF_2$, where ``signature'' is another name for complex-valued constraint, and $\FF_1$ and $\FF_2$ are two sets of signatures. 
Formally, a {\em bipartite Holant problem}, denoted $\holant(\FF_1|\FF_2)$, (on a Boolean domain) is a counting problem defined as follows. 
The problem takes an input instance, called 
a {\em signature grid} $\Omega =(G,\FF'_1|\FF'_2,\pi)$, that consists of a finite undirected bipartite graph $G=(V_1|V_2,E)$ (where all nodes in $V_1$ appear on the left-hand side and all nodes in $V_2$ appear on the right-hand side), 
two {\em finite} subsets $\FF'_1\subseteq \FF_1$ and $\FF'_2\subseteq \FF_2$, and a labeling function $\pi:V_1\cup V_2\rightarrow\FF_1'\cup\FF'_2$ such that  $\pi(V_1)\subseteq \FF'_1$, $\pi(V_2)\subseteq \FF'_2$, and each node $v\in V_1\cup V_2$ is labeled $\pi(v)$, which is a function mapping $\{0,1\}^{deg(v)}$ to $\complex$. 
For convenience, we often write $f_{v}$ for this $\pi(v)$. 
Let $Asn(E)$ denote the set of all {\em edge assignments} $\sigma: E\rightarrow \{0,1\}$. The bipartite Holant problem is meant to compute the complex value $\holant_{\Omega}$: 
\[
\holant_{\Omega} = \sum_{\sigma\in Asn(E)} 
\prod_{v\in V_1\cup V_2}f_{v}(\sigma|E(v)), 
\]    
where $\sigma|E(v)$ denotes the binary string $(\sigma(w_1),\sigma(w_2),\cdots,\sigma(w_k))$ if $E(v)=\{w_1,w_2,\ldots,w_k\}$, sorted in a certain pre-fixed order associated with $f_{v}$. 

Let us define complex-weighted Boolean \#CSPs associated with a set $\FF$ of constraints. Conventionally, a complex-weighted Boolean \#CSP, denoted $\sharpcsp(\FF)$, takes a finite set $\Omega$ 
of ``elements'' of the form $\pair{h,(x_{i_1},x_{i_2},\ldots,x_{i_k})}$ on Boolean variables $x_1,x_2,\ldots,x_n$, where $h\in\FF$ and $i_1,\ldots,i_k\in[n]$. The problem outputs the value $\csp_{\Omega}$: 
\[
\csp_{\Omega} = \sum_{\sigma} \prod_{\pair{h,x'}\in \Omega} h(\sigma(x_{i_1}),\sigma(x_{i_2}),\ldots,\sigma(x_{i_k})),
\] 
where $x'=(x_{i_1},x_{i_2},\ldots,x_{i_k})$ and $\sigma:\{x_1,x_2,\ldots,x_n\}\rightarrow\{0,1\}$ ranges over the set of all {\em variable assignments}. 

Exploiting a close resemblance to Holant problems, we intend to 
adopt the Holant framework and re-define $\sharpcsp(\FF)$ in a form of ``bipartite graphs'' as follows: an input instance to $\sharpcsp(\FF)$ is a triplet $\Omega=(G,X|\FF',\pi)$, which we call a ``constraint frame'' (to distinguish it from the aforementioned conventional framework), where $G$ is an undirected  bipartite graph whose left-hand side contains nodes labeled by Boolean variables and the right-hand side contains nodes labeled by constraints in $\FF'$. 
Throughout this paper, we take this constraint-frame formalism  to treat  complex-weighted Boolean \#CSPs; that is, we always assume that an input instance to  $\sharpcsp(\FF)$ is a certain constraint frame $\Omega$ and an output of $\sharpcsp(\FF)$ is the value $\csp_{\Omega}$. 

The above concept of constraint frame is actually inspired by the fact that  $\sharpcsp(\FF)$ can be viewed as a special case of bipartite Holant problem  $\holant(\{EQ_{k}\}_{k\geq1}|\FF)$ by the following translation: any constraint frame $\Omega$ given to $\sharpcsp(\FF)$ is viewed as a signature grid $\Omega'=(G,\{EQ_k\}_{k\geq1}|\FF',\pi)$ in which each Boolean variable $x$ appearing as a label of node $v$ in the original constraint frame $\Omega$  corresponds to all edges incident on the node $v$ whose label is $EQ_k$ in $G$, and thus each variable assignment for $\Omega$  matches the corresponding 0-1 edge assignment for $\Omega'$. Obviously, each outcome of the constraint frame $\Omega$ coincides with the outcome of the signature grid $\Omega'$.  

To improve readability, we often omit the set notation and express, \eg $\sharpcsp(f,g)$ and $\sharpcsp(f,\FF,\GG)$ to mean $\sharpcsp(\{f,g\})$ and $\sharpcsp(\{f\}\cup \FF\cup\GG)$, respectively.    
When we allow unary constraints to appear in any instance freely, 
we succinctly write $\sharpcspstar(\FF)$ for $\sharpcsp(\FF,\UU)$. In the rest of this paper, we will target the counting problems 
$\sharpcspstar(\FF)$.

\ms
\n{\bf Our Treatment of Complex Numbers.}\hs{1} 
Here, we need to address a technical issue concerning how to handle complex numbers as well as complex-valued functions. Recall that each input instance to a \#CSP  involves a finite set of constraints, which are actually  complex-valued functions. 
How can we compute or manipulate those functions? More importantly, how can we ``express'' them as part of input instances even before starting to compute their values? 

The past literature has exhibited numerous ways to treat complex numbers in an existing framework of theory of string-based computation. There are several reasonable definitions of  ``polynomial-time computable'' complex numbers. They vary depending on which viewpoint we take. To state our results independent of the definitions of computable complex numbers, however, we rather prefer to treat complex numbers as basic ``objects.'' 
Whenever complex numbers are given as part of input instances, we implicitly assume that we have a clear and concrete means of specifying those numbers within a standard framework of computation. Occasionally, however, we will limit our interest within a scope of algebraic numbers, 
as in Lemma \ref{equivalent-plus-star}.  

To manipulate such complex numbers algorithmically, we are limited to 
perform only ``primitive''  operations, such as, multiplications, addition, division, etc., on the given numbers in a very plausible fashion. The execution time of an algorithm that handles those complex numbers is generally measured by the number of those primitive operations. 
To given complex numbers, we apply such primitive operations only; therefore, our assumption on the execution time 
of the operations causes no harm in a later discussion on the computability of $\sharpcsp(\FF)$. (See \cite{CL08,CL11} for further justification.) 

\ms

By way of our  treatment of complex numbers, we naturally define the function class $\fp_{\complex}$ as the set of all complex-valued 
functions that can be computed deterministically on input strings in time 
polynomial in the sizes of the inputs. 

\subsection{Randomized Approximation Schemes}\label{sec:randomized-scheme}

We will lay out a notion of randomized approximation scheme, 
particularly, working on complex numbers. Let 
 $F$ be any counting function mapping from 
$\Sigma^*$ (over an appropriate alphabet $\Sigma$) to $\complex$. Our goal is to approximate each value $F(x)$ when $x$ is given as an input instance to $F$. 
A standard approximation theory (see, \eg \cite{ACG+98}) deals mostly with natural numbers; however, treating complex numbers in the subsequent sections requires an appropriate modification of the standard definition of computability and approximation. In what follows, we will make a specific form of complex-number approximation. 

A fundamental idea behind ``relative approximation error'' is that a maximal ratio between an approximate solution $w$ and a true solution $F(x)$
should be close to $1$.   
Intuitively, a complex number $w$ is an ``approximate solution'' for $F(x)$ if a performance ratio $z=w/F(x)$ (as well as $z=F(x)/w$) is 
close enough to $1$. 
In case when our interest is limited to ``real-valued'' functions, we can expand a standard notion of relative approximation of functions producing non-negative integers (\eg \cite{ACG+98}) and we demand $2^{-\epsilon} \leq  w/F(x)  \leq 2^{\epsilon}$ (whenever $F(x)=0$, we further demand $w=0$). This requirement is logically equivalent to both $2^{-\epsilon} \leq  \left|w/F(x)\right|  \leq 2^{\epsilon}$ and $\arg(F(x))=\arg(w)$ (when $F(x)=0$, $w=0$ must hold), where the ``positive/negative signs'' of real numbers $F(x)$ and $w$ are represented by the ``arguments'' of them in the complex plane. Because our target object is complex numbers $z$, which are always specified by their absolute values  $|z|$ and their arguments $\arg(z)$, both values must be approximated simultaneously.  
Given an {\em error tolerance parameter} $\epsilon\in[0,1]$, we call a value $w$ a {\em $2^{\epsilon}$-approximate solution} for $F(x)$ if $w$ satisfies the following two conditions:
\[
2^{-\epsilon} \leq \left| \frac{w}{F(x)} \right| \leq 2^{\epsilon} 
\hs{5}\text{and}\hs{5} 
\left| \arg\left( \frac{w}{F(x)} \right) \right|\leq \epsilon, 
\]
provided that we apply the following exceptional rule: when $F(x)=0$, we instead require $w=0$.  Notice that this way of approximating  complex numbers is more suitable to establish Lemma \ref{equivalent-plus-star} than the way of approximating both the real parts and the imaginary parts of the complex numbers.

A {\em randomized approximation scheme} for (complex-valued) $F$ is a randomized algorithm that takes a standard input $x\in\Sigma^*$ together with an error tolerance parameter $\varepsilon\in(0,1)$, and outputs a $2^{\epsilon}$-approximate solution (which is a random variable)  for $F(x)$ with probability at least $3/4$. 
A {\em fully polynomial-time randomized approximation scheme} 
(or simply, {\em FPRAS}) for $F$ is a randomized approximation scheme for $F$ that runs in time polynomial in $(|x|,1/\varepsilon)$.

Next, we will describe our notion of approximation-preserving reducibility among counting problems. Of numerous existing notions of approximation-preserving reducibilities (see, \eg \cite{ACG+98}), we choose a notion introduced by Dyer \etalc~\cite{DGGJ03}, which can be viewed as a randomized variant of Turing reducibility, described by a mechanism of {\em oracle Turing machine}. Given two counting functions $F$ and $G$, a {\em polynomial-time (randomized) approximation-preserving (Turing) reduction} (or {\em AP-reduction}, in short) from $F$ to $G$ is a randomized algorithm $N$ that takes a pair $(x,\varepsilon)\in\Sigma^*\times(0,1)$ as input,  
uses an arbitrary randomized approximation scheme (not necessarily polynomial time-bounded) $M$ for $G$ as {\em oracle}, 
and satisfies the following three conditions:
(i) $N$ is a randomized approximation scheme for $F$;  
(ii) every {\em oracle call} made by $N$ is of the form $(w,\delta)\in\Sigma^*\times(0,1)$ satisfying  
$1/\delta \leq p(|x|,1/\varepsilon)$, where $p$ 
is a certain absolute polynomial, 
and an oracle answer is an outcome of $M$ on 
the input $(w,\delta)$; and 
(iii) the running time of $N$ is bounded from above by a certain polynomial in $(|x|,1/\varepsilon)$, not depending on the choice of the oracle $M$. In this case, we write $F\APreduces G$ and we also say that $F$ is {\em AP-reducible} (or {\em AP-reduced}) to $G$. If $F\APreduces G$ and $G\APreduces F$, then $F$ and $G$ are {\em AP-equivalent}\footnote{This concept was called  ``AP-interreducible'' by Dyer \etalc~\cite{DGGJ03} but we prefer this term, which is originated from ``Turing equivalent'' in computational complexity theory.} 
and we write $F\APequiv G$. The following lemma is straightforward.

\begin{lemma}\label{T-con-to-AP-reduction}
If $\FF\subseteq \GG$, then 
$\sharpcspstar(\FF)\APreduces \sharpcspstar(\GG)$. 
\end{lemma}

\section{Underlying Relations and Constraint Sets}\label{sec:constraint-set}

A {\em relation} of arity $k$ is a subset of $\{0,1\}^k$. Such a relation can be viewed as a ``function'' mapping Boolean variables to $\{0,1\}$ (by setting $R(x)=0$ and $R(x)=1$ whenever $x\not\in R$ and $x\in R$, respectively,  for every $x\in\{0,1\}^k$) and it can be treated as a Boolean constraint. For instance, logical relations $OR$, $NAND$, $XOR$, and $Implies$ are all expressed as Boolean constraints in the following manner: 
$OR=[0,1,1]$, $NAND=[1,1,0]$, $XOR=[0,1,0]$, and $Implies=(1,1,0,1)$. The negation of $XOR$ is $[1,0,1]$ and it is simply denoted $EQ$ for convenience. 
Notice that $EQ$ coincides with $EQ_2$. 

For each $k$-ary constraint $f$, its {\em underlying relation} 
is the relation $R_f=\{x\in\{0,1\}^k\mid f(x)\neq0\}$, which 
characterizes the non-zero part of $f$. 
A relation $R$ belongs to the set {\em $IMP$}  (slightly different from $IM_{2}$ in \cite{DGJ10}) 
if it is logically equivalent to a conjunction of  a certain ``positive''  number of relations of the form $\Delta_0(x)$, $\Delta_1(x)$, and $Implies(x,y)$. 
It is worth mentioning that $EQ_2\in IMP$ but $EQ_1\not\in IMP$. 
Moreover, the empty relation ``$\setempty$'' also belongs to $IMP$. 

The purpose of this paper is to extend the scope of 
 the approximation complexity of \#CSPs from Boolean constraints of Dyer \etalc~\cite{DGJ10}, stated in Section \ref{sec:introduction}, to   complex-valued constraints. To simplify later descriptions, it is better for us to introduce the following six special sets of constraints, the first of which has been already introduced in Section \ref{sec:introduction}. 
The notation $f\equiv0$ below means that $f(x_1,\ldots,x_k)=0$ for all $k$-tuples $(x_1,\ldots,x_k)$ in $\{0,1\}^k$, 
where $k$ is the arity of $f$. 


\sloppy
\begin{enumerate}
\item Denote by $\UU$ the set of all unary constraints. 
\vs{-2}
\item Let $\NZ$ be the set of all constraints $f$ of arity $k\geq1$ such that $f(x_1,x_2,\ldots,x_k)\neq0$ for all $(x_1,x_2,\ldots,x_k)\in\{0,1\}^k$. We succinctly call such constraints {\em non-zero constraints}. Notice that this case is different from the case where $f\not\equiv0$. Obviously, $\Delta_0,\Delta_1\not\in\NZ$ holds. 
\vs{-2}
\item Let $\DG$ denote the set of all constraints $f$ of arity $k\geq1$ 
such that $f(x_1,x_2,\ldots,x_k) = \prod_{i=1}^{k}g_i(x_i)$ for certain unary constraints $g_1,g_2,\ldots,g_k$. A constraint in $\DG$ is called {\em degenerate}. Obviously, $\DG$ includes $\UU$ as a proper subset. 
\vs{-2}
\item Define $\ED$ to be the set of functions $f$ of arity $k\geq1$ such that $f(x_1,x_2,\ldots,x_k) = \left(\prod_{i=1}^{\ell_1}h_i(x_{j_i})\right)  \left(\prod_{i=1}^{\ell_2} g_i(x_{m_i},x_{n_i})\right)$ with $\ell_1,\ell_2\geq0$, $\ell_1+\ell_2\geq1$, and $1\leq j_i,m_i,n_i\leq k$, where each $h_i$ is a unary constraint and each $g_i$ is either   
the binary equality $EQ$ or the disequality $XOR$. Clearly, $\DG\subseteq \ED$ holds. The name ``$\ED$'' refers to its key components, ``equality'' and ``disequality.'' See \cite{CLX09} 
for its basic property.
\vs{-2}
\item Let $\IM$ be the set of all constraints $f\not\in\NZ$ of arity $k\geq1$ such that  
$f(x_1,x_2,\ldots,x_k) = \left(\prod_{i=1}^{\ell_1}h_i(x_{j_i})\right)  \left(\prod_{i=1}^{\ell_2} Implies(x_{m_i},x_{n_i})\right)$ with  $\ell_0,\ell_1\geq0$, $\ell_1+\ell_2\geq1$, and $1\leq j_i,m_i,n_i\leq k$, where each $h_i$ is a unary constraint. 
\end{enumerate}

We will present three simple properties of the above-mentioned 
sets of constraints. The first property concerns the set $\NZ$ of non-zero constraints.  
Notice that non-zero constraints will play a quite essential role in 
Lemma \ref{affine-PP-induction} and Proposition \ref{IM-XOR-IMP}.  
In what follows, we claim that two sets $\DG$ and $\ED$ coincide with each other, when they are particularly restricted to non-zero constraints. 
 
\begin{lemma}\label{AG-vs-DD}
Let $f$ be any constraint of arity $k\geq1$ in $\NZ$. It holds that  
$f\in \DG$ iff $f\in\ED$. 
\end{lemma}

\begin{proof}
Let $f$ be any non-zero constraint of arity $k$. 
Note that $f\in \NZ$ iff $|R_f|=2^k$, where $|R_f|$ is the cardinality of the set $R_f$.  
Since $\DG\subseteq \ED$, it is enough to show 
that $f\in\ED$ implies $f\in \DG$. Assume 
that $f$ is in $\ED$. Since $f$ is a product of certain constraints of the forms: $EQ$, $XOR$, and unary constraints. Since $|R_f|=2^k$, $f$ cannot be made of $EQ$ as well as $XOR$ as its ``factors,''   
and thus it should be of the form 
$\prod_{i=1}^{k}U_i(x_i)$, where each $U_i$ is a {\em non-zero} unary constraint. We therefore conclude that $f$ is degenerate and it belongs to  $\DG$. 
\end{proof}


Concerning $\DG$, every constraint $f$ satisfying 
$|R_f|\leq 1$ should belong to $\DG$. This simple fact is shown as follows. 
When $|R_f|=0$, since $f$'s output is always zero,  $f(x_1,\ldots,x_k)$ can be expressed as $\Delta_0(x_1)\Delta_1(x_1)$. Moreover,  when $|R_f|=1$,  assuming that $R_f=\{(a_1,a_2,\ldots,a_k)\}$ for a certain vector $(a_1,a_2,\ldots,a_k)\in\{0,1\}^k$, we set  
$b=f(a_1,a_2,\ldots,a_k)$. Since $f(x_1,x_2,\ldots,x_k)$ equals $b\cdot \prod_{i=1}^{k}\Delta_{a_i}(x_i)$,  $f$ belongs to $\DG$, as requested.  


Several sets in the aforementioned list satisfy the {\em closure property}  under multiplication. For any two constraints $f$ and $g$ of arities $c$ and $d$, respectively, the notation $f\cdot g$ denotes the function defined as follows. For any Boolean vector $(x_1,\ldots,x_k)\in\{0,1\}^k$, let  $(f\cdot g)(x_{m_1},\ldots,x_{m_k}) = f(x_{i_1},\ldots,x_{i_c})g(x_{j_1},\ldots,x_{j_d})$ if $\{m_1,\ldots,m_k\} = \{i_1,\ldots,i_c\}\cup\{j_1,\ldots,j_d\}$, where the order of the indices in   $\{m_1,\ldots,m_k\}$ should be pre-determined from $(i_1,\ldots,i_c)$ and $(j_1,\ldots,j_d)$ before multiplication. For instance, we obtain $(f\cdot g)(x_1,x_2,x_3,x_4)$ from $f(x_1,x_3,x_2)$ and $g(x_2,x_4,x_1)$.   

\begin{lemma}\label{closure-multiplication}
For any two constraints $f$ and $g$ in $\ED$, the constraint $f\cdot g$ is also in $\ED$. A similar result holds for $\DG$, $\NZ$, $\IM$, and $IMP$.
\end{lemma}

\begin{proof}
Assume that $f,g\in\ED$. Note that $f$ and $g$ are both products of constraints, each of which has one of the following forms: $EQ$, $XOR$, unary constraints.  Clearly, the multiplied constraint $f\cdot g$ is a product of those factors, and hence it is in $\ED$. The other cases are similarly proven. 
\end{proof}

{\em Exponentiation} can be considered as a special case of multiplication. To express an exponentiation, we introduce the following notation: for any number $r\in\real-\{0\}$ and any constraint $f$, let $f^r$ denote the  function defined as $f^r(x_1,\ldots,x_n) = (f(x_1,\ldots,x_n))^r$ for any $k$-tuple $(x_1,\ldots,x_k)\in\{0,1\}^k$. 

\begin{lemma}\label{f-vs-f-power-in-PP}
For any number $m\in\nat^{+}$ and any constraint $f$, $f\in\ED$ iff $f^{m}\in\ED$. A similar result holds for $\DG$, $\NZ$, $\IM$, and $IMP$. 
\end{lemma}

\begin{proof}
Let $m\geq1$. Since $f^m$ is the $m$-fold function of $f$, by Lemma \ref{closure-multiplication}, $f\in\ED$ implies $f^m\in\ED$. Next, we intend to  show that $f^m\in\ED$ implies $f\in\ED$. 
Let us assume that $f^m\in\ED$. By setting $g=f^m$, it holds that $f(x_1,\dots,x_n) = (g(x_1,\ldots,x_n))^{1/m}$ for any $n$-tuple  $(x_1,\ldots,x_n)\in\{0,1\}^n$. 
Now, assume that $g = g_1\cdot g_2\cdot\,\cdots\,\cdot g_k$, 
where each $g_i$ is one of $EQ$, $XOR$, and unary constraints.  
If $g_i$ is either $EQ$ or $XOR$, then we define $h_i=g_i$. If $g_i$ is a unary constraint, let us define $h_i = (g_i)^{1/m}$, which is also a unary constraint. Obviously, all $h_i$'s are well-defined and also 
belong to $\ED$, because $\ED$ contains all unary constraints. 
Since $f = h_1\cdot h_2\cdot\,\cdots\,\cdot h_k$, by the definition of $\ED$, we conclude that  
$f$ is in $\ED$. 

The second part of the lemma can be similarly proven.
\end{proof}

\section{Typical Counting Problems}\label{sec:upper-bound}

We will discuss the approximation complexity of 
special counting problems that has arisen naturally in 
the past literature. When we use complex numbers  in the subsequent discussion, we always assume our special way of handling those numbers, as discussed in Section \ref{sec:constraint}. 

The {\em counting satisfiability problem}, \#SAT, is a problem of counting the number of truth assignments that make each given propositional formula true. This problem is proven to 
be complete for $\sharpp$ under AP-reduction \cite{DGGJ03}. 
Dyer \etalc~\cite{DGJ10} further showed that  \#SAT  possesses the computational power equivalent to $\sharpcsp(OR)$ under AP-reduction, namely, 
$\sharpcsp(OR) \APequiv \#\mathrm{SAT}$. 

Nevertheless, to deal particularly with complex-weighted counting problems, 
it is desirable to introduce a complex-weighted version of $\#\mathrm{SAT}$. 
In the following straightforward way, we define  
$\#\mathrm{SAT}_{\complex}$, a complex-weighted version of 
$\#\mathrm{SAT}$. Let $\phi$ be any  propositional formula (with three logical connectives, $\neg$ (not), $\vee$ (or), and $\wedge$ (and))  and 
let $V(\phi)$ be the set of all variables appearing in $\phi$. Let 
$\{w_x\}_{x\in V(\phi)}$ be any series of {\em node-weight functions} 
$w_{x}:\{0,1\}\rightarrow\complex-\{0\}$. 
Given such a pair $(\phi,\{w_x\}_{x\in V(\phi)})$,  $\#\mathrm{SAT}_{\complex}$ asks to compute the sum of all weights 
$w(\sigma)$ for every truth assignment $\sigma$ 
satisfying $\phi$, where $w(\sigma)$ denotes the product of all $w_{x}(\sigma(x))$ for any $x\in V(\phi)$. 
If $w_x(\sigma(x))$ always equals $1$ for every pair of $\sigma$ and $x\in V(\phi)$, then we immediately obtain $\#\mathrm{SAT}$.  This indicates that  $\#\mathrm{SAT}_{\complex}$ naturally extends $\#\mathrm{SAT}$.  

We wish to show that $\sharpcspstar(OR)$ is {\em $\#\mathrm{SAT}_{\complex}$-hard} (\ie at least as hard as $\#\mathrm{SAT}_{\complex}$) under AP-reductions.

\begin{lemma}\label{SAT-to-OR}
$\#\mathrm{SAT}_{\complex}\APreduces \sharpcspstar(OR)$. 
\end{lemma}

The following proof is based on the proof of \cite[Lemma 6]{DGJ10}, which uses approximation results of \cite{DGGJ03} on the {\em counting independent set problem} $\#\mathrm{IS}$.  A set $S$ of nodes in a graph $G$ is called {\em independent} if, for any pair of nodes in $S$, there is no single edge connecting them.
Dyer \etalc~\cite{DGGJ03} showed that  $\#\mathrm{IS}$ 
is AP-equivalent with $\#\mathrm{SAT}$.  
As a complex analogue of $\#\mathrm{IS}$, we introduce  $\#\mathrm{IS}_{\complex}$. An input     
instance to $\#\mathrm{IS}_{\complex}$ is an undirected graph $G=(V,E)$ and a series $\{w_x\}_{x\in V}$  of node-weight functions with each $w_x$ mapping $\{0,1\}$ to $\complex-\{0\}$.  An output of $\#\mathrm{IS}_{\complex}$ is the sum of all weights $w(S)$ for any independent set $S$ of $G$, where $w(S)$ equals the products of all values $w_{x}(S(x))$ over all nodes $x\in V$, where $S(x)=1$ ($S(x)=0$, resp.) iff $x\in S$ ($x\not\in S$, resp.). 

To describe the proof of Lemma \ref{SAT-to-OR}, we wish to introduce a new notation ``$\sharpcspplus(\FF)$,'' which will appear again in Sections \ref{sec:main-theorem} and \ref{sec:elimination}.  
The notation $\sharpcspplus(\FF)$ expresses the counting problem $\sharpcsp(\FF,\UU\cap\NZ)$. 

\begin{proofof}{Lemma \ref{SAT-to-OR}}
We can modify the construction of an AP-reduction from  $\#\mathrm{SAT}$ to $\#\mathrm{IS}$,  
given in \cite{DGGJ03}, by adding a node-weight function to each variable node. Hence, we instantly obtain $\#\mathrm{SAT}_{\complex}\APreduces \#\mathrm{IS}_{\complex}$. We leave the details of the proof 
to the avid reader. Next, we will claim that $\#\mathrm{IS}_{\complex}$ and 
$\sharpcspplus(NAND)$ are AP-equivalent. 
Because this claim is a concrete example of how to relate \#CSPs to  more popular counting problems, here we include the detailed proof of the claim. 

\begin{claim}\label{IS-equiv-NAND}
$\#\mathrm{IS}_{\complex} \APequiv \sharpcspplus(NAND)$.
\end{claim}

\begin{proof}
We want to show that $\#\mathrm{IS}_{\complex}$ is AP-reducible to $\sharpcspplus(NAND)$. Let $G=(V,E)$ and $\{w_x\}_{x\in V}$ be any instance pair to $\#\mathrm{IS}_{\complex}$. In the way described below, we will  construct a constraint frame $\Omega=(G',X|\FF',\pi)$ that becomes 
an input instance to $\sharpcspstar(NAND)$, where $G'=(V|V',E')$ is an  undirected bipartite graph whose $V'$ and $E'$ 
($\subseteq V\times V'$) are defined by the following procedure. Choose any edge $(x,y)\in E$, 
prepare three new nodes $v_1,v_2,v_3$ labeled $NAND, w_{x}, w_{y}$, respectively, and place four edges $(x,v_1),(y,v_1),(x,v_2),(y,v_3)$ into $E'$. 
At the same time, place these new nodes into $V'$. In case where variable $x$ ($y$, resp.) has been already used to insert a new node $v_2$ ($v_3$, resp.), we no longer need to add the node $v_2$ ($v_3$, resp.). 
We define $X$ to be the set of all labels of the nodes in $V$ and define $\FF'$ to be $\{w_{x}\}_{x\in V}\cup \{NAND\}$.  A labeling function $\pi$ is naturally induced from 
$G'$, $X$, and $\FF'$ and we omit its formal description.  

Now, we want to use variable assignments to compute $\csp_{\Omega}$. Given any independent set $S$ for $G$, we define its corresponding variable assignment $\sigma_S$ as follows: for each variable node $x\in V$, let $\sigma_S(x) = S(x)$. 
Note that, for every edge $(x,y)$ in $E$, $x,y\in S$ iff $NAND(\sigma_S(x),\sigma_S(y))=0$.  Let $\tilde{V}$ denote a subset of $V'$ whose elements have the label $NAND$. 
Since all unary constraints appearing as node labels in $V'$  are $w_x$'s, 
$w(S)$ coincides with $\prod_{v\in \tilde{V}}\prod_{x,y\in E'(v)}f_{v}(\sigma_{S}(x),\sigma_{S}(y))\cdot  \prod_{x\in V}w_x(\sigma_{S}(x))$, where each label $f_v$ of node $v$ is $NAND$.  
{}Using this equality, it is not difficult to show that $\csp_{\Omega}$ equals the outcome of $\#\mathrm{IS}_{\complex}$ on the instance $(G,\{w_x\}_{x\in V})$. Therefore, $\#\mathrm{IS}_{\complex}$  is AP-reducible to $\sharpcspplus(NAND)$. 

Next, we will construct an AP-reduction from $\sharpcspplus(NAND)$ to $\#\mathrm{IS}_{\complex}$. Given any input instance $\Omega=(G,X|\FF',\pi)$ with $G=(V_1|V_2,E)$ to $\sharpcspplus(NAND)$, 
we first simplify $G$ as follows. Notice that $\FF'$ is a finite subset of $\{NAND\}\cup\UU$.   
If any two distinct nodes $v_1,v_2\in V_2$ labeled $u_1,u_2\in\UU$, respectively, satisfy $E(v_1)=E(v_2)$, then we merge the two nodes into 
one node with a {\em new} label $u'$, where $u'(x) = u_1(x)u_2(x)$.
Similarly, if $v_1,v_2\in V_2$ with the same label $NAND$ satisfy $E(v_1)=E(v_2)$, then we delete the node $v_1$ and all its incident edges. 
By abusing the notation, we denote the obtained graph by $G$. 

{}From the graph $G$, we define another graph $G'=(V_1,E')$ with $E' = \{(x,y)\in V_1\times V_1 \mid \;\;\text{$\exists v\in V_2$ s.t. $v$ has label $NAND$ and $x,y\in E(v)$}\}$. 
Let $x$ be any variable that appears in $G'$. For each node $w$  in $V_1$ with the label $x$, if $w$ is adjacent to a certain node whose label is a unary constraint, say, $u$, then define $w_x$ to be $u$; otherwise, define $w_x(z)=1$ for any $z\in\{0,1\}$.  Let $\tilde{V}$ be the set of all nodes in $V_2$ whose labels are $NAND$. 
Fix a variable assignment $\sigma$ arbitrarily and define 
$S_{\sigma} = \{ x\in V_1 \mid \sigma(x)=1 \}$. It follows that   
$w(S_{\sigma}) =  \prod_{v\in V_2}f_{v}(\sigma(x_{i_1}),\ldots,\sigma(x_{i_k}))$, where each $k$-tuple $(x_{i_1},\ldots,x_{i_k})$ depends on the choice of $f_v$. 
Thus, $\sum_{\sigma}w(S_{\sigma})$ equals $\csp_{\Omega}$.
Moreover, it holds that  
$\prod_{v\in V_2}f_{v}(\sigma(x_{i_1}),\ldots,\sigma(x_{i_k})) = 
\prod_{v\in \tilde{V}}\prod_{x,y\in E(v)} f_{v}(\sigma(x),\sigma(y)) 
\cdot \prod_{x\in V_1}w_x(\sigma(x))$. Hence,  
$\prod_{v\in V_2}f_{v}(\sigma(x_{i_1}),\ldots,\sigma(x_{i_k})) \neq0$ 
iff $S_{\sigma}$ is an independent set. 
These conditions give the desired AP-reduction from $\sharpcspplus(NAND)$ to $\#\mathrm{IS}_{\complex}$. 
This completes the proof of Claim \ref{IS-equiv-NAND}.
\end{proof}

Naturally, $\sharpcspplus(NAND)$ is AP-reducible to $\sharpcspstar(NAND)$. To complete the proof of the lemma, we want to show that  $\sharpcspstar(NAND)\APreduces \sharpcspstar(OR)$. This is easily shown by, roughly speaking, exchanging the roles of $0$ and $1$ in variable assignments.  More precisely, given an instance $\Omega=(G,X|\FF',\pi)$ to $\sharpcspstar(NAND)$, we build another instance $\Omega'$ by replacing any unary constraint $u$ by $\overline{u}$, where $\overline{u}=[b,a]$ if $u=[a,b]$, and by replacing $NAND$ by $OR$. 
It clearly holds that $\csp_{\Omega} = \csp_{\Omega'}$, and thus $\sharpcspstar(NAND)\APreduces \sharpcspstar(OR)$.  
\end{proofof}

We remark that, by carefully checking the above proof, we can AP-reduce $\#\mathrm{SAT}_{\complex}$ to $\sharpcspplus(OR)$ instead of $\sharpcspstar(OR)$. 
For another remark, we need two new notations. The first notation $\sharpcspplus_{\algebraic}(\FF)$ indicates the counting problem 
obtained from $\sharpcspplus(\FF)$ under the restriction that input instances are limited to algebraic constraints. 
When the outcomes of all node-weight functions of $\#\mathrm{SAT}_{\complex}$ are limited to algebraic complex numbers, we briefly write $\#\mathrm{SAT}_{\algebraic}$. Similar to the first remark, we can prove that $\#\mathrm{SAT}_{\algebraic}$ is AP-reducible to $\sharpcspplus_{\algebraic}(OR)$. This fact will be used in Section \ref{sec:main-theorem}.  

\section{T-Constructibility}\label{sec:constructibility}

One of key technical tools of Dyer \etalc~\cite{DGJ09} in manipulating Boolean constraints is a notion of ``implementation,'' which is used to help establish certain AP-reductions among \#CSPs with Boolean constraints. 
In light of our AP-reducibility, we prefer a more ``operational'' or ``mechanical'' approach toward the manipulation of constraints in a rather  systematic fashion.    
Here, we will present our key technical tool, called {\em T-constructibility}, 
of constructing target constraints from a 
given set of presumably simpler constraints by applying repeatedly such mechanical operations, 
while maintaining the AP-reducibility. This key tool 
will be frequently used in Section \ref{sec:elementary-reduction} 
to establish several AP-reductions among \#CSPs with constraints. 

In an exact counting case of, \eg Cai \etalc~\cite{CLX09x,CLX@,CLX09}, numerous ``gadget'' constructions were used to obtain required properties of constraints. Our systematic approach with the T-constructibility naturally supports most gadget constructions and the results obtained by them can be re-proven by appropriate applications of T-constructibility.  
The set $CL_{T}^{*}(\GG)$ of all constraints that are T-constructed from a fixed set $\GG$ of ``basis'' constraints together with arbitrary free unary constraints is certainly an interesting research object in  promoting our understanding of the AP-reducibility. 
An advantage of taking such a systematic approach can be exemplified, for instance, by Lemma \ref{IMP-substitute}, in which we are able to argue the {\em closure property} under AP-reducibility (without the projection operation). This property is a key to the subsequent lemmas and propositions.  
This line of study was lately explored in \cite{BDGJ12}.  

To pursue notational succinctness, we use the following notations  in the rest of this paper. For any index $i\in[k]$ 
and any bit $c\in\{0,1\}$, 
let the notation $f^{x_i=c}$ denote the function $g$ satisfying that $g(x_1,\ldots,x_{i-1},x_{i+1},\ldots,x_k) = f(x_1,\ldots,x_{i-1},c,x_{i+1},\ldots,x_k)$ for any vector $(x_1,\ldots,x_{i-1},x_{i+1},\ldots,x_k)\in\{0,1\}^{k-1}$.  
Similarly, for any two {\em distinct}  
indices $i,j\in[k]$, we denote by $f^{x_i=x_j}$ the function $g$ defined as  $g(x_1,\ldots,x_{i-1},x_{i+1},\ldots,x_k) =
 f(x_1,\ldots,x_{i-1},x_{j},x_{i+1},\ldots,x_k)$. 
Moreover, let $f^{x_i=*}$ be the function $g$ defined as  $g(x_1,\ldots,x_{i-1},x_{i+1},\ldots,x_k) = \sum_{x_i\in\{0,1\}}f(x_1,\ldots,x_{i-1},x_i,x_{i+1},\ldots,x_k)$, where $x_i$ is no longer a free variable.  
By extending these notations naturally, we can write, \eg $f^{x_i=0,x_{m}=*}$ as the shorthand for $(f^{x_i=0})^{x_m=*}$ and $f^{x_i=1,x_m=0}$ for $(f^{x_i=1})^{x_m=0}$.  

We say that a constraint $f$ of arity $k$ is {\em T-constructible} (or {\em T-constructed}) from a constraint set $\GG$ if $f$ can be obtained, initially from constraints in $\GG$, by applying recursively a finite number (possibly zero) of functional operations described below.
\begin{enumerate}\vs{-1}
\item {\sc Permutation:} for two indices $i,j\in[k]$ with $i<j$, 
by exchanging two columns $x_i$ and $x_j$ in  $(x_1,\ldots,x_i,\ldots,x_j,\ldots,x_k)$, transform $g$ into $g'$ that is 
defined as $g'(x_1,\ldots,x_i,\ldots,x_j,\ldots,x_k) = g(x_1,\ldots,x_j,\ldots,x_i,\ldots,x_k)$.
\vs{-2}
\item {\sc Pinning:} for an index $i\in[k]$ and a bit $c\in\{0,1\}$, build $g^{x_i=c}$ from $g$.
\vs{-2}
\item {\sc Projection:} for an index $i\in[k]$, build $g^{x_i=*}$ 
from $g$.
\vs{-2}
\item {\sc Linking:} for two {\em distinct} indices $i,j\in[k]$, build $g^{x_i = x_j}$ from $g$.
\vs{-2}
\item {\sc Expansion:} for an index $i\in[k]$, introduce a new ``free'' variable, say, $y$ and  transform $g$ into $g'$ that is defined by $g'(x_1,\ldots,x_i,y,x_{i+1},\ldots,x_k) =  g(x_1,\ldots,x_{i},x_{i+1},\ldots,x_k)$.
\vs{-2}
\item {\sc Multiplication:} from two constraints $g_1$ and $g_2$ of 
arity $k$ sharing the same input variable series $(x_1,\ldots,x_k)$, 
build  $g_1\cdot g_2$, that is,  $(g_1\cdot g_2)(x_1,\ldots,x_k) = g_1(x_1,\ldots,x_k)g_2(x_1,\ldots,x_k)$. 
\vs{-2}
\item {\sc Normalization:}  for a constant $\lambda\in\complex-\{0\}$, build $\lambda\cdot g$ from $g$, where $\lambda\cdot g$ is defined as $(\lambda \cdot g)(x_1,\ldots,x_k) = \lambda\, g(x_1,\ldots,x_k)$. 
\end{enumerate}\vs{-1}
When $f$ is T-constructible from $\GG$, we write $f\leq_{con}\GG$. In particular, when $\GG$ is a singleton, say, $\{g\}$, we also write $f\leq_{con}g$ instead of $f\leq_{con}\{g\}$ for succinctness.  With this notation $\leq_{con}$, an earlier notation $CL_{T}^{*}(\GG)$ can be formally defined as the set $CL_{T}^{*}(\GG) = \{f\mid f\leq_{con}\GG\cup\UU\}$. 

As is shown below, T-constructibility induces a partial order among  all constraints. The proof of the following lemma is rather straightforward, 
and thus we omit it entirely and leave it to the avid reader. 
 
\begin{lemma}\label{AP-transitivity}
For any three constraints $f$, $g$, and $h$, it holds 
that (i) $f\leq_{con}f$ and (ii)  $f\leq_{con}g$ and $g\leq_{con}h$ imply $f\leq_{con}h$. 
\end{lemma}

The usefulness of T-constructibility comes from the following lemma, 
which indicates the invariance of T-constructibility under AP-reductions.  
For readability, we place the proof of the lemma in Section \ref{sec:elimination}. 

\begin{lemma}\label{constructibility}
If $f\leq_{con}\GG$, then $\sharpcspstar(f,\FF)\APreduces \sharpcspstar(\GG, \FF)$  
for any set $\FF$ of constraints.  
\end{lemma}

\section{Expressive Power of Unary Constraints}\label{sec:elementary-reduction} 

In the rest of this paper, we aim at proving our dichotomy theorem (Theorem \ref{dichotomy-theorem}). Its proof, which will appear in Section \ref{sec:main-theorem}, is comprised of several crucial ingredients. 
A starting point of the proof of the dichotomy theorem  
is a tractability lemma of \#CSP$^*$s---Lemma \ref{basic-case-FPC}---which states that $\sharpcspstar(\FF)$ is solvable in polynomial-time if  $\FF\subseteq\ED$.  A free use of arbitrary unary constraint plays an essential role in this section.


Notice that, whenever $\FF\subseteq \ED$,  $\sharpcspstar(\FF)$ coincides with $\sharpcsp(\FF)$ since $\UU\subseteq\ED$. As part of their dichotomy theorem stated in Section \ref{sec:introduction}, Cai \etalc~\cite{CLX09x,CLX09} demonstrated that, for any constraint set $\FF$ included in $\ED$ (denoted $\PP$ in \cite{CLX09x,CLX09}),  $\sharpcsp(\FF)$ can be solved in polynomial time. {}From this tractability result, Lemma \ref{basic-case-FPC} immediately follows. For completeness, nevertheless, we will briefly sketch an outline of the proof of the lemma. 

For the proof, we need to 
consider ``factorization'' of a given constraint $g$.  Let us recall that, 
 when $g$ is in $\ED$,  $g$ can be expressed as a multiplication of the form $g_1\cdot g_2\cdot\,\cdots\,\cdot g_n$, where each $g_i$ is one of $EQ$, $XOR$, and unary constraints. For convenience, we call the list $L=\{g_1,g_2,\ldots,g_n\}$ of all those factors  a {\em factor list} for $g$.  Since $g\leq_{con}L$ holds, it follows by Lemma \ref{constructibility}, that $\sharpcspstar(g,\FF)\APreduces \sharpcspstar(L,\FF)$ for any constraint set $\FF$. 


\begin{lemma}\label{basic-case-FPC}
For any constraint set $\FF$, if $\FF\subseteq \ED$, then $\sharpcspstar(\FF)$ is in $\fp_{\complex}$. 
\end{lemma}

\begin{proofsketch}
Consider any constraint frame $\Omega=(G,X|\FF',\pi)$ given as an input instance to $\sharpcspstar(\FF)$, where $G=(V_1|V_2,E)$ is an bipartite undirected graph and $\FF'$ is a finite subset of $\FF\cup\UU$. 
Since $\FF\subseteq \ED$ by the premise of the lemma, we assume that $\FF'\subseteq \ED$. 
Thus, it is possible to replace every constraint in $\Omega$ by its ``factors'' so that we can assume that $G$ is composed of nodes whose labels are limited to $EQ$, $XOR$, and unary constraints. 

We next modify the graph $G$ as follows. For each node $v$ labeled $EQ$, we merge into a single node any two nodes in $V_1$ that are adjacent to $v$ and we then delete $v$ as well as its incident edges. 
After this deletion, we assume that there is no node with the label $EQ$. 
Furthermore, if two nodes $v_1$ and $v_2$ both labeled $XOR$ are adjacent to the same nodes in $V_1$, then we delete the node $v_2$ and its incident edges. Hereafter, we assume that no such node pair of $v_1$ and $v_2$ exists.   
To compute $\csp_{\Omega}$, it suffices to consider 
every constraint frame $\Omega'$ obtained from $\Omega$ by restricting its scope within a connected component $G'=(V'_1|V'_2,E')$ consisting only of nodes whose labels are $XOR$ or unary constraints.  

Toward the value $\csp_{\Omega'}$, we first select a special node, say, $v$ in $V'_1$ in the following fashion. If there exists a cycle that contains all nodes in $V'_1$, then we choose any node in $V'_1$ as $v$. Otherwise, we choose a node $v\in V'_1$ that is not adjacent to any two nodes having the label $XOR$. Let $x$ be the ``variable'' label of this 
node $v$. It is enough to focus on a Boolean value of this particular variable $x$ since  Boolean values of the remaining  variables are automatically induced by the choice of the value of $x$. Hence, $\csp_{\Omega'}$ is computed simply by assigning only two values 
($0$ or $1$) to $x$.  
\end{proofsketch}


Henceforth, we will focus our 
attention on the remaining case where $\FF\nsubseteq \ED$.
As a basis to the subsequent analysis, the rest of this section is 
devoted to explore fundamental  
properties of binary constraints $f$ and it shows numerous complexity bounds of $\sharpcspstar(f)$'s. A key to our study is an expressive power 
of free unary constraints.  
We begin with a quick reminder that, since all unary constraints are free to use,  it obviously holds that $\sharpcspstar(\Delta_0,\Delta_1,\FF)\APequiv \sharpcspstar(\FF)$ for any set $\FF$ of constraints.
 
Earlier, in the proof of Lemma \ref{SAT-to-OR}, we have demonstrated the AP-equivalence between $\sharpcsp(OR)$ and $\sharpcsp(NAND)$ by a simple technique of swapping the roles of $0$ and $1$. However, this technique is not sufficient to prove that $\sharpcspstar(OR,\FF)\APequiv \sharpcspstar(NAND,\FF)$ for an arbitrary constraint set $\FF$. A use of 
unary constraint, on the contrary, helps us establish 
this stronger AP-equivalence.

\begin{proposition}\label{NAND-OR}
For any constraint set $\FF$, $\sharpcspstar(OR,\FF) \APequiv \sharpcspstar(NAND,\FF)$. 
\end{proposition}

\begin{proof}
We will show only one direction of $\sharpcspstar(OR,\FF)\APreduces \sharpcspstar(NAND,\FF)$, since the opposite direction is similarly proven. 
For brevity, let $f=NAND$ and set $u=[1,-1]$. 
Now, we claim that $OR\leq_{con}\{f,u\}$. 
For this purpose, let us define $g(x_1,x_2) = 
\sum_{x_3\in\{0,1\}}f(x_1,x_3)f(x_3,x_1)u(x_3)$. It is not difficult 
to show that $g$ equals $OR=[0,1,1]$. Hence, $OR$ is T-constructed from 
$\{f,u\}$, as requested. 
{}From this T-constructibility, by Lemma \ref{constructibility}, we obtain an AP-reduction from  
$\sharpcspstar(OR,\FF)$ to $\sharpcspstar(f,u,\FF)$. 
The last term obviously equals 
$\sharpcspstar(f,\FF)$ because $u$ is a unary constraint.   
Therefore, we conclude that $\sharpcspstar(OR,\FF)\APreduces \sharpcspstar(f,\FF)$. 
\end{proof}

Now, let us consider other binary constraints. 
Of them, our next target is constraints having the forms:  
$(0,a,b,1)$ or $(1,a,b,0)$ with $ab\neq0$. 

\begin{lemma}\label{0-b-c-case}
Let $a,b\in\complex$ with $ab\neq0$ and let $f$ be any constraint of the form:  either $(0,a,b,1)$ or $(1,a,b,0)$. 
For any constraint set $\FF$, the following statement holds:  
$\sharpcspstar(OR,\FF) \APreduces \sharpcspstar(f,\FF)$. 
\end{lemma}

\begin{proof}
Consider the case where $f=(0,a,b,1)$. Let $u=[1,ab]$ for brevity.  
We want to claim that $OR$ is T-constructed from the 
constraint set $\{f,u\}$. To show this claim, define $g(x_1,x_2) = f(x_1,x_2)f(x_2,x_1)u(x_1)u(x_2)$. A simple calculation leads us to the conclusion that  $g=[0,a^2b^2,a^2b^2]$. By normalizing $g$ appropriately, 
we immediately obtain another constraint 
$g'=[0,1,1]$, which clearly equals $OR$.  By the definition of $g$, 
it thus follows that $OR\leq_{con}\{f,u\}$. Lemma \ref{constructibility} then implies that $\sharpcspstar(OR,\FF)$ is AP-reducible to $\sharpcspstar(f,u,\FF)$. Since $u$ is unary, the last term coincides with  $\sharpcspstar(f,\FF)$, yielding the desired consequence of the lemma.  

For the case of $f=(1,a,b,0)$, a similar argument shows that 
$\sharpcspstar(NAND,\FF) \APreduces \sharpcspstar(f,\FF)$. 
By Proposition  \ref{NAND-OR}, 
it is possible to replace $NAND$ by $OR$, and therefore 
the desired consequence follows.  
\end{proof}

We will examine other binary constraints of the form $(0,a,b,0)$ with $ab\neq0$. 

\begin{lemma}\label{OR-XOR-bounds}
Let $a,b\in\complex$ with $ab\neq 0$.  
For any set $\FF$ of constraints,  
$\sharpcspstar(XOR,\FF)\APreduces \sharpcspstar((0,a,b,0),\FF)$ holds.
\end{lemma}

\begin{proof}
Let $f=(0,a,b,0)$ with $ab\neq0$. Now, we ``symmetrize'' $f$ by setting $g(x_1,x_2)= f(x_1,x_2)f(x_2,x_1)$, which yields the equation   $g=[0,ab,0]$. We normalize $g$ and then obtain $[0,1,0]$, 
which is exactly $XOR$. We thus conclude that $XOR\leq_{con} f$, 
implying that 
$\sharpcspstar(XOR,\FF)$ is AP-reducible to $\sharpcspstar(f,\FF)$ 
by Lemma \ref{constructibility}.
\end{proof}

Next, we move our interest to the special relation 
$Implies$. Unlike the case of unweighted Boolean \#CSPs \cite{DGJ09}, where it remains open whether $\sharpcsp(Implies)$ is AP-equivalent to $\sharpcsp(OR)$, a heavy use of non-zero unary constraints leads to a surprising AP-equivalence between $\sharpcspstar(Implies,\FF)$ and $\sharpcspstar(OR,\FF)$ for any constraint set $\FF$.  

\begin{proposition}\label{implies-vs-OR}
For any constraint set $\FF$, it holds that $\sharpcspstar(Implies,\FF) \APequiv \sharpcspstar(OR,\FF)$. 
\end{proposition}

This proposition directly follows from the lemma stated below together with Lemma \ref{constructibility}, which translates T-constructibility into AP-reducibility. 

\begin{lemma}
\begin{enumerate}
\item There exists a finite set $\GG\subseteq\UU$ such that 
$Implies\leq_{con}\GG\cup \{OR\}$. 
\item There exists a finite set $\GG\subseteq \UU$ such that 
$OR \leq_{con} \GG\cup \{Implies\}$.   
\end{enumerate}
\end{lemma}

\begin{proof}
For ease of the description that follows, we set $f= Implies$.

(1) Here, we intend to claim the T-constructibility of $f$ from the set $\{OR,u_1,u_2,u_3\}$, where  $u_1=[1,-1/2]$, $u_2=[2,-2/3]$, and $u_3=[1,-1/8]$. 
We will prove this claim by building a series of 
{\em T-constructible} constraints.  
First, we define $g(x,y) = \sum_{z\in\{0,1\}} OR(x,z) OR(z,y)$. This implies that $g=[1,1,2]$ and $g\leq_{con}OR$. Next, let $h(x,y)$ be $\sum_{z\in\{0,1\}} g(x,z)g(z,y)g(y,z) u_1(z)$, which equals  $(1/2,-1,0,-3)$. 
Clearly, it holds that $h\leq_{con}\{g,u_1\}$. 
Moreover, let $h'(x,y) = h(x,y) u_2(x)$, implying $h'=(1,-2,0,-2)$. 
Finally, we set 
$p(x,y) = \sum_{z\in\{0,1\}} h'(x,z) h'(z,y) h'(y,z) u_3(z)$. 
A simple calculation shows that $p=(1,1,0,1)$. 
Since $p$ is T-constructible from $\{h',u_3\}$, we then obtain $f\leq_{con}\{OR,u_1,u_2,u_3\}$, as requested. 

(2) We will show that $OR$ is T-constructed from the set $\{f,\Delta_0,u_1,u_2\}$, where $u_1=[1,-8]$ and $u_2=[49,24]$. 
To prove this claim, we introduce the following two useful constraints: $h_2=[2,1,1]$ and $h_3=[2,1,1,1]$. In what follows, we will prove 
(i) $OR \leq_{con} \{h_2,u_1,u_2\}$ and (ii) $h_2\leq_{con}\{f,\Delta_0\}$. {}From Statements (i) and (ii), it immediately follows by Lemma \ref{AP-transitivity} that $OR\leq_{con} \{f, \Delta_0, u_1, u_2\}$, as requested. 

(i) We start with defining 
$g(x,y) = \sum_{z\in\{0,1\}} h_2(x,z) h_2(z,y) h_2(y,z) u_1(z)$.    
It is easy to check that $g=(0,-6,-4,-7)$. With this $g$, we 
define $s(x,y) = g(x,y)g(y,x)u_2(x)u_2(y)$, which equals 
$(0,a,a,a)$, where $a = 28224$. 
By normalizing $s$ properly, we immediately obtain the constraint $OR$. 
Therefore, it holds 
that $OR\leq_{con} \{h_2,u_1,u_2\}$. 

(ii) We note that $EQ_3$ is T-constructed from $f$ because  
$EQ_3(x,y,z)$ equals $f(x,y) f(y,z) f(z,x)$. 
Using $EQ_3$, we define
\[
p(x_1,y_1,z_1) = \sum_{x_2,y_2,z_2\in\{0,1\}} EQ_3(x_2,y_2,z_2) f(x_1,x_2) f(y_1,y_2) f(z_1,z_2).
\]
This definition implies that  
$p = [2,1,1,1]$, and thus $p$ equals $h_3$. This means that  
$h_3 \leq_{con} \{EQ_3,Implies\}$. 
Since $h_2 = h_3^{x_1=0}$, $h_2$ is T-constructed from $\{h_3,\Delta_0\}$. 
{}From all the obtained results, we easily conclude that $h_2 \leq_{con}  \{f,\Delta_0\}$. 
\end{proof}

Next, we target constraints of the forms $(1,a,0,b)$ and $(1,0,a,b)$ with $ab\neq0$. 

\begin{lemma}\label{arity-2-implies-reduction}
Let $f=(1,a,0,b)$ with $a,b\in \complex$. If $ab\neq 0$, then  $\sharpcspstar(OR,\FF)\APreduces \sharpcspstar(f,\FF)$  holds 
for any constraint set $\FF$. 
By permutation as well as normalization, $(1,0,a,b)$ also yields the same consequence. 
\end{lemma}

\begin{proof}
Let $f=(1,a,0,b)$ with $ab\neq0$. 
In this proof, we use two unary constraints: 
$u=[1,a/b]$ and $v=[1,1/a^3]$. With a help of Proposition 
\ref{implies-vs-OR}, our goal is now set to show the T-constructibility of  $Implies$ from $\{f,u,v\}$. Firstly, by defining  $g(x_1,x_2) = f(x_1,x_2)u(x_1)$,  we obtain $g=(1,a,0,a)$. This implies that $g\leq_{con}\{f,u\}$. Secondly, we define $h(x_1,x_2) = \sum_{x_3\in\{0,1\}} g(x_1,x_3)g(x_3,x_2)g(x_2,x_3)v(x_3)$. A simple calculation shows that $h=(1,1,0,1)$.  This concludes that $Implies$ is T-constructed from $\{g,v\}$. By combining those two results, we obtain $Implies\leq_{con}\{f,u,v\}$ by Lemma \ref{AP-transitivity}. The desired result then follows immediately because $u$ and $v$ are unary constraints. 
\end{proof}

Now, we consider the case of constraints $f$ having the form $(1,x,y,z)$ 
and demonstrate the hardness of $\sharpcspstar(f)$. For complex-valued constraints $f$, this case is quite special because, if they are Boolean, they  all become $[1,1,1]$ and fall into $\DG$. 

\begin{proposition}\label{1-x-y-z-case}
Let $x,y,z\in\complex$. If both $xyz\neq0$ and $xy\neq z$ hold, 
then $\sharpcspstar(OR,\FF) \APreduces 
\sharpcspstar((1,x,y,z),\FF)$ 
for any set $\FF$ of constraints.
\end{proposition}

When $xy=z$, on the contrary, the constraint $f=(1,x,y,z)$ becomes degenerate, since $f(x_1,x_2)$ equals $[1,y](x_1)\cdot [1,x](x_2)$.  Proposition \ref{1-x-y-z-case} is a direct consequence of two lemmas---Lemmas \ref{1-x-y-z-implies} and \ref{1-x-y-xy-OR}---each of which handles a different case.   
Let us begin with the case where $xyz\neq0$ and $xy\neq\pm{z}$. 
 
\begin{lemma}\label{1-x-y-z-implies}
Let $x,y,z\in\complex$. If $xyz\neq0$ and $xy\neq \pm{z}$, 
then $\sharpcspstar(OR,\FF) \APreduces 
\sharpcspstar((1,x,y,z),\FF)$ 
for any set $\FF$ of constraints.
\end{lemma}

\begin{proof}
Let $f=(1,x,y,z)$ with $xyz\neq0$.  
Assuming that $xy\neq \pm{z}$, we first want to show that 
 $\sharpcspstar(Implies,\FF)$ is AP-reducible to $\sharpcspstar(f,\FF)$. 
With an unknown variable $a$, set $u=[1,a]$.  Now, we define 
$g(x_1,x_2) = \sum_{x_3\in\{0,1\}}f(x_1,x_3)f(x_3,x_2)f(x_2,x_3)u(x_3)$. 
A simple calculation provides an equation $g = (1+ax^2y,x(y+az^2),y(1+axz),xy^2+az^3)$. 
By setting $a$ to be $-1/xz$,  
we obtain $g=(1-xy/z,x(y-z/x),0,xy^2-z^2/x)$, which implies  
$g \leq_{con} \{f,u\}$. It thus follows that    
$\sharpcspstar(g,\FF)\APreduces \sharpcspstar(f,\FF)$ by Lemma \ref{constructibility}. 
Note that three entries in $g$ are non-zero, since $xy\neq z$ and $x^2y^2\neq z^2$. 
Apply Lemma \ref{arity-2-implies-reduction} to a normalized $g$. 
As an immediate consequence, we obtain an AP-reduction from  
$\sharpcspstar(OR,\FF)$ to $\sharpcspstar(g,\FF)$. 
The final result is obtained by combining the two AP-reductions. 
\end{proof}

Finally, we consider the remaining case where $xy=-z$; that is, constraints of the form $(1,x,y,-xy)$, which is excluded in the previous lemma. 

\begin{lemma}\label{1-x-y-xy-OR}
Let $x,y\in\complex$. If $xy\neq0$, then $\sharpcspstar(OR,\FF)\APreduces \sharpcspstar((1,x,y,-xy),\FF)$ for any constraint set $\FF$. 
\end{lemma}

\begin{proof}
Our proof strategy is to reduce this case to Lemma \ref{1-x-y-z-implies}. Let $f = (1,x,y,-xy)$ and assume that $xy\neq0$. Define $u=[1,a]$ and   consider the constraint $g$ defined by $g(x_1,x_2) = \sum_{x_3\in\{0,1\}} f(x_1,x_3)f(x_3,x_2)u(x_3)$. This $g$ satisfies  $g=(1+axy,x(1-axy),y(1-axy),xy(1+axy))$. If we choose $a=2/xy$, 
then we have $g=(3,-x,-y,3xy)$, which equals $3\cdot(1,-x/3,-y/3,xy)$. 
For simplicity, set $x'=-x/3$, $y'=-y/3$, and $z'=xy$. 
Note that $x'$, $y'$, and $z'$ are all non-zero. Now, we set $h=(1,x',y',z')$ that is obtained by normalizing $g$. 
Since $x'y'\neq \pm{z'}$, we can apply Lemma \ref{1-x-y-z-implies} 
to this $h$ and the desired consequence then follows. 
\end{proof}

\section{Useful Properties of Specific Constraints}\label{sec:support-reduction}

We have shown in the previous section numerous complexity bounds of $\sharpcspstar(f)$'s  when $f$'s are of arity $2$. Our next step is to 
show similar bounds of $\sharpcspstar(f)$'s for constraints $f$ of higher arities. To achieve our goal, we first explore fundamental properties of constraints related to $\ED$, $\IM$, and $\NZ$ so that those properties will contribute to proving the desired hardness results in Section \ref{sec:IM2-support}.  


Underlying relations of constraints $f$ play a 
distinguishing role in our analysis of the behaviors of the counting problems $\sharpcspstar(f)$. In particular, basic properties of relations in $IMP$ become a crucial part of the proof of our dichotomy theorem.  
Let us recall that a relation $R$ in $IMP$ is expressed as a 
product of the constant constraints as well as $Implies$. 
To handle relations in $IMP$, it is convenient to introduce a notion of  ``imp support.''   
A constraint $f$ is said to have {\em imp support} if $R_f$ is in $IMP$. 
It is not difficult to show that all constraints in $\IM$ have imp support.  The converse also holds for any binary constraint. 

\begin{lemma}\label{binary-IM-IMP}
For any binary constraint $f$, it holds that $f\in\IM$ iff $R_f\in IMP$. 
\end{lemma}

\begin{proof}
Since the underlying relation of any constraint $f$ in $\IM$ belongs to $IMP$, it suffices to show that if $R_f\in IMP$ then $f\in\IM$. Assume that $R_f$ is in $IMP$. Depending on the form of $R_f$, we consider two cases separately. 

(i) Consider the case where $f$ has the form $(x,y,0,z)$ with $x,y,z\in\complex$ and $y\neq0$.  It is easy to check that $f(x_1,x_2)$ always equals $[y,z](x_1)[x/y,1](x_2) Implies(x_1,x_2)$. Thus, $f$ should belong to $\IM$. 

(ii) Next, we consider the case where $f$ has the form $(x,0,0,z)$ with $x,z\in\complex$.  Obviously, $f(x_1,x_2)$ always coincides with $[x,z](x_1) Implies(x_1,x_2) Implies(x_2,x_1)$.  This shows that $f$ is in $\IM$. 
\end{proof}

The first useful property is a {\em closure property} under a certain restricted case of T-constructibility. In what follows, we will show that the T-constructibility without the projection operation preserves 
the membership to $\ED$ and the property of imp support; in other words, the set $\ED$ as well as the set of all constraints that have imp support is closed under T-constructibility with no projection operation.  

\begin{lemma}\label{IMP-substitute}
Let $f$ be any constraint and let $\GG$ be any constraint set. 
Assume that $f$ is T-constructible from $\GG$ using 
no projection operation. 
\begin{enumerate}\vs{-1}
\item If all constraints in $\GG$ have imp support,  
then $f$ also has imp support.
\vs{-2}
\item If all constraints in $\GG$ are in $\ED$, then $f$ is also in $\ED$. 
\end{enumerate}
\end{lemma}

In Section \ref{sec:elementary-reduction}, we have defined the concept of ``factor list'' for a given constraint in $\ED$. Similarly, for a relation $R$ in $IMP$, we can define its ``factor list'' using its factors of the forms, $\Delta_0$, $\Delta_1$, and $Implies$. 

\begin{proofof}{Lemma \ref{IMP-substitute}}
Let $f$ be any $k$-ary constraint and let $\GG$ be any constraint set. 
Assume that $f$ is T-constructed from $g$ 
(or $\{g_1,g_2\}$ in the case of the multiplication operation) in $\GG$ 
by a single application of one of the operations described in Section \ref{sec:constructibility} except for the projection operation. 
The proof proceeds by induction on the number of operations that are applied to T-construct $f$ from $\GG$. Clearly, the basis case (\ie $f\in\GG$) is trivial. 

(1) Assume that $g$ has imp support and let $L$ be a factor list for $R_g$.   
We aim at proving that $R_f$ is in $IMP$ by modifying this factor 
list $L$ step by step. Because the cases for the operations of 
normalizing, permutation, and expansion are trivial, we will concentrate on the remaining operations. For ease of notational complication, we are focused on specific indices in the following argument. 

[{\sc Pinning}] Let us consider the case $f= g^{x_i=0}$. To keep our proof clean, we set $i=1$ without loss of generality. Notice that $R_f = R_{g}^{x_1=0}$. Now, we need to eliminate all occurrences of $x_1$ from $L$.  For any index $j\in[k]$, if there is a factor $Implies(x_1,x_j)$ in $L$, then we delete it from the list. 
If a factor $Implies(x_j,x_1)$ exists in $L$, then we replace it by $\Delta_0(x_j)$. If $L$ contains a factor $\Delta_0(x_1)$, then we simply delete it from $L$. Finally, if there exists a factor $\Delta_{1}(x_1)$ in $L$, then we choose any variable, say, $x_2$ appearing in $f$ and define $L'$ to be $\{\Delta_0(x_2),\Delta_1(x_2)\}$ since $f\equiv0$ and $k\geq1$.  Clearly, the obtained list, say, $L'$ lacks any entry of $x_1$. Since $L'$ preserves all the factors associated with the remaining variables, $L'$ should be a factor list for $R_f$. Therefore,  $f$ has imp support. In a similar manner, we can handle the case of $g^{x_1=1}$.

[{\sc Linking}] Let $f= g^{x_i=x_j}$. For simplicity, we set $i=1$ 
and $j=2$. In the factor list $L$, we replace all occurrences of $x_1$ by $x_2$. 
For instance, if $L$ has a factor of the form $Implies(x_1,x_3)$, then we replace it with $Implies(x_2,x_3)$. The newly obtained list becomes a factor list for $R_f$, and thus $R_f$ belongs to $IMP$ since $R_g$ is in $IMP$.  

[{\sc Multiplication}] Finally, assume that $f=g_1\cdot g_2$. We denote by $L_1$ and $L_2$ two factor lists for $R_{g_1}$ and $R_{g_2}$, respectively. 
We combine these two lists into the union $L_1\cup L_2$, which becomes a factor list for $R_f$. Therefore, $f$ has imp support.   

(2) The proof for $\ED$ is in essence similar to (1); in particular, the multiplication and the linking operations are treated almost identically. Here, we note only a major difference. In the case of the pinning operation, say, $f=g^{x_1=0}$, if there exists a factor of the form $EQ(x_1,x_j)$ ($XOR(x_1,x_j)$, resp.) in a factor list $L$ for $R_g$, then we replace it by $\Delta_0(x_j)$ ($\Delta_1(x_j)$, resp.). This manipulation eliminates the variable $x_1$ from the list $L$, and thus the resulting list becomes a factor list for $R_f$.   
\end{proofof}

For any constraint $f$ having imp support, by its definition, 
its underlying relation $R_f$ can be factorized as  
$R_f = g_1\cdot g_2\cdot\,\cdots\,\cdot g_{m}$, where 
each factor $g_i$ is one of the following forms:  
$\Delta_0(x)$, $\Delta_1(x)$, and $Implies(x,y)$ 
($x$ and $y$ may be the same). 
The factor list $L =\{g_1,g_2,\ldots,g_{m}\}$ for $R_f$    
is said to be {\em imp-distinctive}\footnote{This notion is called ``normalized'' in \cite{DGJR09}; however, we have already used the term ``normalization'' in a different context.} if (i) no single variable appears 
both in $\Delta_{c}$ and $Implies$ in $L$, where $c\in\{0,1\}$, and 
(ii) no factor of the form $Implies(x,x)$ belongs to $L$. 
In Lemma \ref{distinctive-list}, we will show that such an imp-distinctive list always 
exists for an arbitrary constraint $f$ with imp support although 
such a list may not be {\em unique} in general.   

\begin{lemma}\label{distinctive-list}
For each constraint $f$ having imp support, there exists an  imp-distinctive 
factor list for $R_f$. 
\end{lemma}

\begin{proof}
This proof is similar to the proof of \cite[Lemma 4]{DGJR09}. Let $f$ be any constraint that has imp support. 
Let $L$ be any factor list for $R_f$, composed of  relations of the forms $\Delta_0(x)$, $\Delta_{1}(x)$, and $Implies(x,y)$, where $x$ and $y$ are appropriate variables.    
Since this list $L$ may not be imp-distinctive in general, we need to 
run the following five processes repeatedly to make $L$ imp-distinctive.  

(i) {}From the factor list $L$, delete all factors of the form 
$Implies(x,x)$.  After this process, we assume that $L$ contains no such factor. 
(ii) If $\{\Delta_{0}(x),Implies(x,y)\}\subseteq L$, then delete $Implies(x,y)$. (iii) If $\{\Delta_{0}(y),Implies(x,y)\}\subseteq L$, then replace  $Implies(x,y)$ in $L$ by $\Delta_{0}(x)$. (iv) If $\{\Delta_{1}(x),Implies(x,y)\}\subseteq L$, then replace the factor  $Implies(x,y)$ in $L$ by $\Delta_{1}(y)$. (v) If $\{\Delta_{1}(y),Implies(x,y)\}\subseteq L$, then delete  $Implies(x,y)$ from $L$.   

Assume that no process is further applicable to the obtained factor list, say, $L'$. We want to make a claim that $L'$ is indeed imp-distinctive. 
Toward a contradiction, let us assume otherwise. Suppose that there exists a variable $x$ appearing in both $\Delta_c$ (where $c\in\{0,1\}$) and $Implies$ in $L'$. When $c=0$, we can further apply either Process (ii) or Process  (iii)  to $L'$. This is a contradiction against the definition of $L'$. The case of $c=1$ is similar. Next, suppose that a variable $x$ appears in $Implies(x,x)$ in $L'$. In this case, Process (i) can be applied to $L'$, a contradiction. Therefore, it follows that $L'$ is imp-distinctive. 
\end{proof}


We have utilized a certain form of ``factorization'' of constraints.   
In fact, most constraints $f$ can be expressed as products of a finite number of certain types of ``factors,'' which are usually ``simpler'' than the original constraints. 
Here, we look for particular factorization that is obtained by factors of the following forms: $\Delta_0(x)$, $\Delta_1(x)$, and $EQ(x,y)$. 
After dividing $f$ by those factors, the remaining portion of the constraint can be described by a notion of ``simple form.'' To explain this notion, we need to introduce new terminology.   
For each constraint $f$ of arity $k$, its {\em representing Boolean matrix} $M_{f}$ is composed of rows indexed by all instances $a=(a_1,a_2,\ldots,a_k)$ in $R_f$ (in the standard lexicographical order) and columns  indexed by numbers in $[k]$, and 
each $(a,i)$-entry of $M_{f}$ is a Boolean value $a_i$. 
We say that a constraint is in {\em simple form} 
if its representing Boolean matrix does not contain all-$0$ columns, all-$1$ columns, or any pair of identical columns. 
Clearly, any all-$0$ constraint $f$ (\ie $f\equiv0$) cannot be in simple form. 

As shown in Lemma \ref{simple-form}, it is always possible to factorize any given constraint $f$ into two factors, at least one of which must be in simple form. For the proof of this lemma, we will deal with a representing Boolean matrix $M_f$ of the constraint $f$ and we will execute a {\em sweeping procedure} that  eliminates, one by one, unwanted columns of $M_f$ 
until the remaining matrix becomes a simple form. 
The lemma will become useful in the proof of Proposition 
\ref{no-affine-and-IM2}. Recall that $EQ_1$ represents $[1,1]$. 

\begin{lemma}\label{simple-form}
Let $f$ be any constraint of arity $k\geq1$. If $f$ never belongs to $(IMP\cap\ED)\cup\{EQ_1\}$ after any normalization, then there exist two indices $m$ and $m'$ with $1\leq m\leq m'\leq k$, an arity-$m'$ relation $R$ in $(IMP\cap \ED)\cup \{EQ_1\}$, and a constraint $g$ of arity $k-m$ such that (after properly permuting variable indices) $f(x_1,\ldots,x_k) = R(x_1,\ldots,x_{m'})g(x_m,\ldots,x_k)$, $m\neq k$, $g\leq_{con}f$, and $g$ is in simple form. Moreover, $f$ has imp support iff $g$ has imp support, and $f\in\ED$ iff  $g\in\ED$. 
\end{lemma}

\begin{proof}
Let $f$ be any constraint of arity $k\geq1$. Assume that $f$ cannot equal $c\cdot R'$ for a certain constant $c\in\complex$ and a certain relation $R'$ in $(IMP\cap\ED)\cup\{EQ_1\}$.  
To generate a constraint of simple form, we run the following 
algorithm, called a {\em sweeping procedure}. The algorithm uses two parameters $g$ and $R$, and it updates them at each step   
until $g$ becomes the desired simple form. 
Initially, we set $g$ to be $f$ and set $R$ to be $EQ_1$ over a single variable, say, $x_1$. Suppose below that, after an appropriate re-ordering of variable indices,  $g$ 
and $R$ have the forms $g(x_{e},x_{e+1},\ldots,x_k)$ and $R(x_1,x_2,\ldots,x_d)$ for two numbers $d$ and $e$ satisfying $1\leq e\leq d\leq k$. 
Let $M_g$ denote the representing Boolean matrix of $g$.     

(i) Assume that there exists an all-$0$ column indexed, say, $i$ in $M_g$.   
We then delete this column $i$ from $M_g$. When this situation happens, 
$g(x_e,\ldots,x_k)$ must be factorized into $\Delta_0(x_i)  g^{x_i=0}(x_e,\ldots,x_{i-1},x_{i+1},\ldots,x_k)$. After the deletion of the column $i$, we update $g$ to be $g^{x_i=0}$ 
and set $R$ to be $\Delta_0\cdot R$ that is defined as  
$(\Delta_0\cdot R)(x_1,\ldots,x_d,x_i) = 
\Delta_0(x_i) R(x_1,\ldots,x_d)$ if $d<i$, and $(\Delta_0\cdot R)(x_1,\ldots,x_d) = \Delta_0(x_i) R(x_1,\ldots,x_d)$ if $i\leq d$.  
(ii) If an all-$1$ column, say, $i$ exists in $M_g$, then we delete the column $i$.  Since $g(x_e,\ldots,x_k) = \Delta_1(x_i)  g^{x_i=1}(x_e,\ldots,x_{i-1},x_{i+1},\ldots,x_k)$, we can update $g$ 
and $R$ to $g^{x_i=1}$ and $\Delta_1\cdot R$, 
respectively.  
(iii) Assuming that there are no all-$0$ and all-$1$ columns, 
if there is a pair of identical columns, say, $i$ and $j$ ($i<j$), then we delete the column $i$. 
Note that  $g(x_e,\ldots,x_k)$ equals $EQ(x_i,x_j)  g^{x_i=x_j}(x_e,\ldots,x_{i-1},x_{i+1},\ldots,x_{j},\ldots,x_k)$. After the deletion, we update $g$ and $R$ respectively to $g^{x_i=x_j}$ and 
$EQ\cdot R$, where  
$(EQ\cdot R)(x_1,\ldots,x_d,x_i,x_j)= EQ(x_i,x_j)R(x_1,\ldots,x_d)$ 
if $d<i$,   $(EQ\cdot R)(x_1,\ldots,x_d,x_j)= EQ(x_i,x_j)R(x_1,\ldots,x_{i-1},x_j,x_{i+1},\ldots,x_d)$ if $i\leq d<j$, 
and $(EQ\cdot R)(x_1,\ldots,x_d)= EQ(x_i,x_j)R(x_1,\ldots,x_d)$ if $j\leq d$. 

After an execution of the above sweeping procedure, we obtain a relation $R$ and a constraint $g$ satisfying the equation $f(x_1,\ldots,x_k) = R(x_1,\ldots,x_{m'}) g(x_m,\ldots,x_k)$ (after an appropriate permutation of variable indices). 
In particular, when $m=1$, none of the cases (i)--(iii) occurs, and thus $R$ equals $EQ_1$ and $g$ coincides with $f$. When $m\neq1$, the procedure guarantees that 
$R$ is a product of some of $\Delta_0$, $\Delta_1$, and $EQ$. In either case, $R$ belongs to $(IMP\cap\ED)\cup\{EQ_1\}$. 
Now, we will show that $m\neq k$. If $m=k$, then $g$ has no input variable. 
Thus, $g$ becomes a constant, say, $c$ in $\complex$, which implies that $f$ equals $c\cdot R$. Obviously, $f$ belongs to $(IMP\cap\ED)\cup\{EQ_1\}$ after appropriate normalization. 
This is a clear contradiction against our assumption, and therefore we conclude that $m\neq k$.  
Moreover, $g$ should be in simple form. 
The procedure clearly ensures that $g\leq_{con}f$. 

The second part of the lemma is shown as follows. 
Assume that $f$ has imp support. By the behavior of the sweeping procedure, 
$g$ should be  T-constructible from $f$ {\em without} the projection operation. 
Lemma \ref{IMP-substitute}(1) then ensures that $g$ has  
imp support as well. 
Finally, we will show that if $g$ has imp support then $f$ has imp support. As a starting point, assume that $g$ has imp support. 
Notice that, since $f=R\cdot g$, $R_f$ coincides with $R\cdot R_g$. Since both $R_g$ and $R$ are in $IMP$, Lemma \ref{closure-multiplication} shows that  $R\cdot R_g$ belongs to $IMP$. Therefore, $f$ has imp support.  In a similar fashion, it is easy to show, using Lemma \ref{IMP-substitute}(2), that $f\in\ED$ iff $g\in\ED$.  
\end{proof}

Non-zero constraints require a special attention for weighted \#CSPs because their underlying relations are all equal to $\{0,1\}^k$ (where $k$ is the arities of the constraints) and they cannot be dealt with simply by a  factorization technique.  
However, we can show that every non-degenerate constraint in $\NZ$  T-constructs a quite useful binary constraint residing in $\NZ$. 

\begin{lemma}\label{affine-PP-induction}
Let $f\in \NZ$ be any constraint of arity $k\geq 2$ 
and let $\FF$ be any set of constraints.  
If $f\not\in\DG$, then there exists a constraint 
$h=(1,x,y,z)$ for which $xyz\neq 0$, $z\neq xy$, 
$h\leq_{con}f$. In particular, $h$ (before normalizing) has the form $f^{x_3=c_3,\ldots,x_k=c_k}$ for certain constants $(c_3,\ldots,c_k)\in\{0,1\}^{k-2}$, after an appropriate permutation of variable indices.  
\end{lemma}

\begin{proof}
This proof is part of the proof of \cite[Lemma 4.4]{CLX09} meant 
for exact counting of weighted \#CSPs with complex-valued constraints. 
A similar argument for non-negative constraints is found in the proof of 
\cite[Lemma 14]{DGJ09}. Since the proof of this lemma is not difficult, for completeness,  we include the proof.  

Assume that $f$ is a non-zero constraint of arity $k\geq2$.  
For each index $i\in[k]$, using the assumption $f\in\NZ$, we define a constraint $g_i$ as $g_i(x_1,\ldots,x_{i-1},x_{i+1},\ldots,x_k) = f^{x_i=1}(x_1,\ldots,x_{i-1},x_{i+1},\ldots,x_k)/f^{x_i=0}(x_1,\ldots,x_{i-1},x_{i+1},\ldots,x_k)$. Let us show by contradiction the existence of an index $i\in[k]$ such that $g_i$ is not a constant function.  
Suppose otherwise. For each index $i\in[k]$, we define $u_i =[1,b_i]$, where 
$b_i\in\complex$ is a unique constant satisfying  $g_i(x_1,\ldots,x_{i-1},x_{i+1},\ldots,x_k)=b_i$ for any vector $(x_1,\ldots,x_{i-1},x_{i+1},\ldots,x_k)\in\{0,1\}^{k-1}$.  Such $b_i$ actually exists because $g_i$ is a constant function. 
Since $f^{x_1=1}(x_2,\ldots,x_k) = b_1 f^{x_1=0}(x_2,\ldots,x_k)$, 
it follows that $f(x_1,\ldots,x_k) = u_1(x_1) f^{x_1=0}(x_2,\ldots,x_k)$. By a similar argument, we also have $f(x_1,\ldots,x_k) = u_1(x_1)u_2(x_2)f^{x_1=0,x_2=0}(x_3,\ldots,x_k)$. After repeating this argument, in the end, we obtain $f(x_1,\ldots,x_k) = f^{x_1=0,\ldots,x_k=0}()\prod_{i=1}^{k}u_i(x_i)$, where  $f^{x_1=0,\ldots,x_k=0}()$ is a certain complex number. This indicates that $f$ is degenerate, and thus it belongs to $\DG$, a contradiction. 
Therefore, for a certain index $i\in[k]$, $g_i$ is not a constant function. 
Set $i=1$ for simplicity. 

Let us choose a sequence $(a_3,\ldots,a_k)\in\{0,1\}^{k-2}$ for which  $g_1(0,a_3,\ldots,a_k)\neq g_1(1,a_3,\ldots,a_k)$. Define $h = f^{x_3=a_3,\ldots,x_k=a_k}$, which must have the form $(w,x,y,z)$. {}From  $h\in\NZ$, $xyzw\neq0$ follows immediately. By normalizing $h$  appropriately, we can assume that $h$ takes the form of $(1,x,y,z)$ with $xyz\neq0$. Moreover, from  $h(1,0)/h(0,0)\neq h(1,1)/h(0,1)$, we  obtain the inequality $xy\neq z$. 
\end{proof}

\section{Imp Support and the Hardness of \#CSPs}\label{sec:IM2-support}

Based on various properties given in Sections \ref{sec:constructibility}--\ref{sec:support-reduction}, we will present, in Propositions  \ref{no-affine-and-IM2} and \ref{IM-XOR-IMP}, two hardness results on the approximation complexity of  
certain counting problems $\sharpcspstar(f)$. 
We will show these results by 
building appropriate AP-reductions from $\sharpcspstar(OR)$, indicating  
 the $\#\mathrm{SAT}_{\complex}$-hardness of the  $\sharpcspstar(f)$'s by Lemma \ref{SAT-to-OR}. Moreover, those results look ``complementary;'' that is, Proposition \ref{no-affine-and-IM2} deals with constraints having imp support whereas Proposition \ref{IM-XOR-IMP} targets constraints lacking imp support. They will become a core of the proof of our main theorem in Section \ref{sec:main-theorem}. 


Our first focal point is to discuss constraints that have imp support. Particularly, we are interested in the case where 
the constraints are not in $\ED$. 

\begin{proposition}\label{no-affine-and-IM2}
Let $f$ be any constraint having imp support and 
let $\FF$ be any constraint set. 
If $f\not\in\ED$, then 
$\sharpcspstar(OR,\FF)\APreduces \sharpcspstar(f,\FF)$. 
\end{proposition}

The proof of this proposition requires five claims concerning 
the relation $Implies$.  
We begin with a useful result, shown in \cite{DGJ10},  on  
relations residing outside of $IMP\cup\NZ$.
For its description, we introduce two additional notations. For any two vectors $a=(a_1,\ldots,a_k)$ and $b=(b_1,\ldots,b_k)$ in $\{0,1\}^k$, the notation $a\wedge b$ denotes the vector 
$(a_1\wedge b_1,\ldots,a_k\wedge b_k)$ and $a\vee b$ denotes $(a_1\vee b_1,\ldots,a_k\vee b_k)$, where $a_i\wedge b_i = \min\{a_i,b_i\}$ and $a_i\vee b_i = \max\{a_i,b_i\}$. 

\begin{lemma}\label{DGJ10-IMP-property}{\rm \cite[Corollary 18]{DGJ10}}\hs{1}
For any nonempty relation $R\not\in IMP\cup\NZ$, there are two distinct 
instances $a$ and $b$ in $R$ such that either 
$a\wedge b \not\in R$ or $a\vee b \not\in R$ holds.
\end{lemma}

Secondly, we present a simple characterization of binary constraints not in $\IM\cup\NZ$. 

\begin{lemma}\label{arity-2-IM-support}
For any constraint $f$ of arity $2$ with 
$f\not\equiv0$, $f\not\in \IM\cup\NZ$  
iff $f$ is of the form $(a,b,c,d)$ with $ad=0$ and $bc\neq 0$. 
\end{lemma}

\begin{proof}
Let $f$ be any binary constraint satisfying that $f\not\equiv0$. Since $f$ is binary, by Lemma \ref{binary-IM-IMP}, it holds that $f\not\in\IM$ iff $R_f\not\in IMP$. 

(Only If--part) Assume that $f$ does not belong to $\IM\cup\NZ$.  
Since $R_f\not\in IMP$, Lemma \ref{DGJ10-IMP-property} guarantees the existence of two distinct elements $a'=(a_1,a_2)$ and $b'=(b_1,b_2)$ in $R_f$ satisfying 
that either $a'\wedge b'\not\in R_f$ or $a'\vee b'\not\in R_f$. Let us consider the first case where $a_1=b_1=0$. Since $a'\neq b'$, we obtain  $a_2\neq b_2$ and thus 
$\{a'\wedge b',a'\vee b'\} = \{a',b'\} \subseteq R_f$, a contradiction. The 
second case $a_1=b_1=1$ is similar. Consider the third case 
$a_1\neq b_1$. Without loss of generality, we set $a_1=0$ and $b_1=1$. 
When  $a'=(0,0)$, we obtain both $a' = a'\wedge b'$ 
and $b' = a'\vee b'$, 
leading to a contradiction. When  $a'=(0,1)$,  $b'$ should equal $(1,0)$ because, otherwise,  $a' = a'\wedge b'$ 
and $b' = a'\vee b'$  follow. Since either $a'\wedge b'\not\in R_f$ or $a'\vee b'\not\in R_f$, $R_f$'s outcome should be one of the following three forms: $(0,1,1,1)$, $(1,1,1,0)$, and $(0,1,1,0)$. In other words, $R_f$ (seen as a function) equals $OR$, $NAND$, or $XOR$. {}From this consequence, the lemma immediately follows. 

(If--part) Let $f=(a,b,c,d)$ with $a,b,c,d\in\complex$ and  
assume that $ad=0$ and $bc\neq0$. This instantly implies $f\not\in\NZ$. 
Next, we wish to show that $f\not\in\IM$. Toward a contradiction, 
assume that  $f$ is in $\IM$, implying $R_f\in IMP$.  
Lemma \ref{distinctive-list} then yields 
an  imp-distinctive factor  list for $R_f$. Such a list should be a subset of $\{Implies(x_1,x_2),Implies(x_2,x_1),\Delta_0(x_j),\Delta_1(x_j) \mid j=1,2\}$. 
Let us consider all possible imp-distinctive lists for $R_f$. By checking them carefully, we can find that all the lists define only $13$ binary relations, excluding $OR$, $NAND$, and $XOR$.  In addition, it is not difficult to show that, for any binary relation $R\not\in\{OR,NAND,XOR\}$, if $R=R_f$ then either $ad\neq0$ or $bc=0$ holds. This clearly contradicts our assumption on $f$. Therefore, we reach a conclusion that $f\not\in\IM$.  
\end{proof}

As an immediate corollary of Lemma \ref{arity-2-IM-support}, we obtain a 
characterization of binary constraints in $\IM$. Recall that $\IM\cap\NZ=\setempty$. 

\begin{corollary}\label{aryty-2-IM-abcd}
For any constraint $f=(a,b,c,d)$ with $a,b,c,d\in\complex$ 
and $f\not\equiv0$, $f\in\IM$ iff $bc=0$. 
\end{corollary}

\begin{proof}
Let $f=(a,b,c,d)$ with $a,b,c,d\in\complex$ and $f\not\equiv0$.  
Lemma \ref{arity-2-IM-support} states that $f\in \IM\cup\NZ$ 
iff either $ad\neq 0$ or $bc=0$ holds. Obviously, $f\in\IM$ implies $f\not\in\NZ$. 
Moreover, it holds that $f\not\in \NZ$ iff $abcd=0$. Because if $ad\neq 0$ and $abcd=0$ then at least one of $b$ and $c$ should be $0$, the corollary follows immediately.  
\end{proof}

The third claim is more technical. To explain it, we need to introduce a directed graph $G_{f,L}$ induced from a factor list $L$ for $R_f$. The graph $G_{f,L}$ consists of nodes whose names are variables $x_{i_1},x_{i_2},\ldots,x_{i_k}$ appearing in $R_f(x_{i_1},x_{i_2},\ldots,x_{i_k})$ and of edges $(x,y)$ whenever a factor $Implies(x,y)$ is in $L$. We call this $G_{f,L}$ an {\em imp graph} of $R_f$ and $L$.  
We say that a factor list $L$ for $R_f$ is {\em good} if  (i) $L$ consists only of $Implies$'s, (ii) every node in $G_{f,L}$ is adjacent to at least one node in $G_{f,L}$, and (iii) there is no cycle in $G_{f,L}$. Note that, whenever $f$ has a good factor list,  Condition (iii) prohibits $f$ from belonging to $\ED$.   

The notation $COMP_1(f)$ for a constraint $f$ of arity $k$  
means the set $\{f^{x_i=c}\mid i\in[k], c\in\{0,1\}\}$. Furthermore, let $COMP_2(f)=\{f^{x_i=c,x_j=d}\mid i,j\in[k], i\neq j, c,d\in\{0,1\}\}$. 
Every constraint in $COMP_{1}(f)\cup COMP_{2}(f)$ is obviously 
T-constructible from $f$ by applications of the pinning operation.  

\begin{lemma}\label{good-factor-list}
For any constraint $f$ of arity $k\geq3$, if $R_f$ has a good factor list, then there exists a constraint $h\in COMP_1(f)\cup COMP_2(f)$ such that $R_h$ has a good factor list. 
\end{lemma}

\begin{proof}
Let $R_f$ be the underlying relation of an arity-$k$ constraint $f$ defined on $k$ Boolean variables $\{x_1,\ldots,x_k\}$. 
Let $L$ be any good factor list for $R_f$ and let $G_{f,L}=(V,E)$ be an imp graph of $R_f$ and $L$ with $V=\{x_1,\ldots,x_k\}$. There are two cases to handle differently. 

(1) Suppose that there exists an index $i\in[n]$ for which  (i) 
$(x_i,x_j)\not\in E$ holds for all indices $j\in[k]-\{i\}$ and (ii)  the incident set $E(x_i)$ of the node $x_i$ is a singleton. By the property of the imp graph, a certain index $j\in[k]-\{i\}$ must satisfy that $(x_j,x_i)\in E$. Since $|E(x_i)|=1$, this node $x_j$ should be unique. Now, we are focused on this particular node $x_j$. 

(a) Assume that $|E(x_j)|>1$.  For the desired $h$ stated in the lemma, we set $h=f^{x_i=1}$, which belongs to $COMP_1(f)$. 
Next, we want to show that $R_h$ has a good factor list. Let us define $L'=L-\{Implies(x_j,x_i)\}$. It is easy to show that $L'$ is a factor list for $R_h$.  With this list $L'$, define $G_h$ to be an imp graph of $R_h$ and $L'$. Note that $G_{h,L'}$ contains no node named $x_i$. 
Obviously, every node in $G_{h,L'}$ is adjacent to at least one node in $G_{h,L'}$.  
Moreover, there is no cycle in $G_{h,L'}$ because any cycle in $G_{h,L'}$  becomes a cycle in $G_{f,L}$.  Therefore, $L'$ is a good factor list for $R_h$. 

(b) Assume that $|E(x_j)|=1$. This means $E(x_j)=\{x_i\}$, and the graph $H=(\{x_i,x_j\},\{(x_j,x_i)\})$ forms a connected component of $G_{f,L}$.  
Here, we set $h=f^{x_i=1,x_j=1}$ so that $h$ belongs to $COMP_2(f)$. We define $L'=L-\{Implies(x_j,x_i)\}$, which becomes a factor list for $R_h$. 
Note that $L'$ cannot be empty because, otherwise, $L$ consists only of $Implies(x_j,x_i)$ and thus $k=2$ follows, a contradiction.  
Now, we claim that $L'$ is good. Let $G_{h,L'}$ be an imp graph of $R_h$, which has neither the node $x_i$ nor the node $x_j$. Note that every node 
in $G_{h,L'}$ is adjacent to at least one node because deleting the  
subgraph $H$ does not affect the adjacency property of the other nodes 
in $G_{f,L}$.  Thus, $L'$ is a good factor list for $R_h$.  

(2) Assume that Case (1) does not happen. Choose a variable $x_i$ so that $(x_j,x_i)\not\in E$ for any $j\in [k]-\{i\}$.  Such a variable should exist because there is no cycle in $G_{f,L}$. The desired $h$ is now defined as 
$h= f^{x_{1}=0}$, which clearly falls into $COMP_1(f)$.  
Let $L' = L-\{Implies(x_i,x_j)\mid j\in[k]-\{i\}\}$.  
This $L'$ becomes a factor list for $R_h$. If any node $x_j$ with $j\neq i$ is deleted from $G_{f,L}$, then $|E(x_j)|=1$ follows and this $x_j$ satisfies Case (1). This is a contradiction; hence,  an imp graph of $R_h$ and $L'$ lacks only the node $x_i$. This ensures that the properties of $L$ are naturally inherited to $L'$; therefore, $L'$ is good. 
\end{proof}

The notion of good factor list is closely related to that of simple form. Exploring this relationship, we can prove the following corollary, in which we decrease the arity of a given constraint while maintaining the imp-support  property and the non-membership property to $\ED$. 

\begin{corollary}\label{support-IM-not-affine}
Let $f$ be any constraint of arity $k\geq3$. 
Assume that $f$ is in simple form. 
If $f$ has imp support and $f\not\in\ED$, then there exists a constraint $h$ of arity less than $k$ 
such that  $h$ has imp support, $h\not\in\ED$, and $h\leq_{con}f$.
\end{corollary}

\begin{proof}
Let $k\geq3$ and let $f$ be any arity-$k$ constraint having imp support. 
Assume that $f$  is in simple form but it is not in $\ED$.  
Since $f$ has imp support,  by Lemma \ref{IMP-substitute}(1), 
every constraint $h$ that is T-constructed from $f$ by applications of the pinning operation has imp support. Since $R_f\in IMP$,  
we choose an imp-distinctive factor list $L$ for $R_f$. 
Note that every factor in $L$ is of the form $Implies$ because if $L$ contains a factor $\Delta_c$, where $c\in\{0,1\}$, then $M_f$ must contain an all-$c$ column, a contradiction against the simple-form property of $f$.  

To appeal to Lemma \ref{good-factor-list}, we need to show that $L$ is a good factor list for $R_f$. Let $G_{f,L}$ denote an imp graph of $R_f$ and $L$. 
Firstly, we deal with a situation where there exists a variable that appears in no factor in $L$. We choose such a variable, say, $x_i$ and define $h= f^{x_i=0}$, which clearly belongs to $COMP_1(f)$. Moreover, the value of $x_i$ does not affect the computation of $f$; thus, it follows that  $f^{x_i=0}(x_1,\ldots,x_{i-1},x_{i+1},\ldots,x_k) = f^{x_i=1}(x_1,\ldots,x_{i-1},x_{i+1},\ldots,x_k)$.   
Therefore, we obtain $f(x_1,\ldots,x_k) = h(x_1,\ldots,x_{i-1},x_{i+1},\ldots,x_k)$, which  implies that $h$ has imp support. 
Since $f\not\in\ED$, we conclude that $h\not\in \ED$. 
Hereafter, we assume that every variable appears in 
at least one factor in $L$.  

Secondly, we will show that $G_{f,L}$ has no cycle. Suppose otherwise; namely, for a certain series $(x_{i_1},x_{i_2},\ldots,x_{i_m})$  of variables, the set  
$\{Implies(x_{i_j},x_{i_{j+1}}), Implies(x_{i_m},x_{i_1}) \mid j\in[m-1]\}$ is included in $L$. Clearly, $M_f$ includes two identical columns $i_1$ and $i_m$,   
and thus $f$ cannot be in simple form, a contradiction. Therefore, 
no cycle exists in $G_{f,L}$.  
Overall, we conclude that $L$ is a good factor list for $R_f$. 
Lemma \ref{good-factor-list} then gives a constraint $h$ such that $h$ is T-constructed from $f$ by one or more applications of the pinning operation  and $R_h$ has a good factor list. Therefore, $h$ should have imp support. Moreover, the definition of good factor list implies that $h$ should not belong to $\ED$. The use of the pinning operation guarantees that 
the arity of $h$ should be less than $k$. 
\end{proof}


Finally, we will give the proof of Proposition \ref{no-affine-and-IM2}.

\begin{proofof}{Proposition \ref{no-affine-and-IM2}}
Let $f$ be any constraint of arity $k\geq1$ and let $\FF$ be any constraint set.   
We will show by induction on $k$ that if $f$ has imp support but it is   
not in $\ED$ then  
$\sharpcspstar(OR,\FF)$ is AP-reduced to $\sharpcspstar(f,\FF)$. 
Let us assume that $f$ has  imp support and $f\not\in\ED$. Note that $f\not\equiv0$ because, otherwise, $f$ belongs to $\ED$. 

[Basis Case: $k=1$] In this case, the proposition is trivially true, because all unary functions are already in $\ED$.  

[Next Case: $k=2$] Assume that $f=(a,b,c,d)$ with 
$a,b,c,d\in\complex$.  Since $f$ is a binary constraint,   the imp-support property of $f$ makes $f$ belong to $\IM$.  Since $f\not\equiv0$, 
Corollary \ref{aryty-2-IM-abcd} yields  $bc=0$. 
Now, we examine the following three possible cases. 

(1) The first case is that $b=0$ but $c\neq0$. Let us examine all four possible values of $f$. Write $u$ for the constraint $[c,d]$. (i) If 
$f=(0,0,c,0)$, then $f$ is clearly in $\ED$. (ii) Let $f=(0,0,c,d)$ with $d\neq0$. The value $f(x_1,x_2)$ actually equals $\Delta_{1}(x_1)u(x_2)$, 
and thus $f$ belongs to $\ED$. (iii) If $f=(a,0,c,0)$ with $a\neq0$, then 
$f$ has the form $u(x_1)\Delta_0(x_2)$, implying $f\in\ED$. These three cases immediately lead to a contradiction against the assumption $f\not\in\ED$. 
(iv) The remaining case is that $f=(a,0,c,d)$ with $ad\neq0$. 
By normalizing $f$ appropriately, we may assume that $f$ has the form $(1,0,c,d)$. Now,  
we apply Lemma \ref{arity-2-implies-reduction} and then obtain the desired AP-reduction from $\sharpcspstar(OR,\FF)$ to $\sharpcspstar(f,\FF)$. 

(2) The second case where $b\neq0$ and $c=0$ is symmetric to Case (1) and is omitted. 

(3) Let us consider the third case where $b=c=0$.  There are four possible choices for $f$: (i') $f=(a,0,0,0)$ with $a\neq0$, (ii') $f=(0,0,0,d)$ with $d\neq0$, (iii') $f=(a,0,0,d)$ with $ad\neq0$, and (iv') $f=(0,0,0,0)$. In all those four cases, clearly $f$ belongs to $\ED$, a contradiction. This completes the case of $k=2$.  

[Induction Case: $k\geq3$] 
As the induction hypothesis, we assume that the proposition is true 
for any constraint of arity less than $k$. 

(1) Assume that $f$ falls into $(IMP\cap\ED)\cup\{EQ_1\}$ after appropriate normalization; in other words, $f$ equals $c\cdot R$, where $c\in\complex$  and $R\in (IMP\cap\ED)\cup\{EQ_1\}$. There are two cases,  $R\in IMP\cap\ED$ or $R=EQ_1$, to consider. In either case, however, $f$ belongs to $\ED$. This contradicts our assumption. 

(2) Assume that Case (1) does not occur. 
Lemma \ref{simple-form} then provides a relation $R$ in 
$(IMP\cap \ED)\cup\{EQ_1\}$   
and a constraint $g$ in simple form 
that satisfy $g\leq_{con}f$ and $f= R\cdot g$. Moreover, the second part of Lemma \ref{simple-form} implies that $g$ has imp support and $g$ does not belong to $\ED$.  
Firstly,  we consider the case where $R\neq EQ_1$. Since $f\neq g$, an execution 
of the sweeping procedure given in the proof of Lemma \ref{simple-form} makes the arity of $g$ smaller than that of $f$.  
The induction hypothesis therefore implies that $\sharpcspstar(OR,\FF)\APreduces \sharpcspstar(g,\FF)$. 
Since $g\leq_{con}f$, by Lemmas \ref{AP-transitivity} and \ref{constructibility}, we obtain $\sharpcspstar(OR,\FF)\APreduces \sharpcspstar(f,\FF)$.  
Secondly, we consider the case of $R= EQ_1$. This case implies  
$f=g$,  and thus $f$ should be in simple form. Appealing to   
Corollary \ref{support-IM-not-affine}, we obtain a 
constraint $h$ of arity smaller than $k$ satisfying that $h\leq_{con}f$, $h\not\in\ED$, and $h$ has imp support.  Our induction hypothesis 
then ensures  that $\sharpcspstar(OR,\FF)\APreduces \sharpcspstar(h,\FF)$ holds. Moreover, since $h\leq_{con}f$, 
$\sharpcspstar(h,\FF)$ is AP-reduced to $\sharpcspstar(f,\FF)$ by Lemma \ref{constructibility}. The desired conclusion of the proposition follows by combining those two  AP-reductions. 
\end{proofof}


In Proposition \ref{no-affine-and-IM2}, 
we have discussed constraints with imp support. Our second focal point is to  discuss constraints that {\em lack} imp support, provided that they are chosen from the outside of $\ED\cup\NZ$.   

\begin{proposition}\label{IM-XOR-IMP}
Let $f$ be any constraint not in $\ED\cup\NZ$. If $f$ has no imp support, 
then $\sharpcspstar(OR,\FF)\APreduces \sharpcspstar(f,\FF)$ for any constraint set $\FF$. 
\end{proposition}

The proof of this proposition relies on Lemma \ref{arity-2-IM-support}, which gives  a complete characterization of binary  
constraints inside $\IM\cup\NZ$. The proposition is proven easily by an  
assist of Lemma \ref{IMP-substitute} as well.  

\begin{proofof}{Proposition \ref{IM-XOR-IMP}}
Let $f$ be any constraint of arity $k\geq1$ and assume that $f$ has no imp support and $f\not\in\ED\cup\NZ$.  
Our proof proceeds by induction on $k$. 
The base case $k=1$ is  trivial since all unary constraints belong to 
$\ED$.  Next, assume that $k=2$. Notice that $f$ cannot be in $\IM$ since $R_f\not\in IMP$ by Lemma \ref{binary-IM-IMP}. We then  apply Lemma \ref{arity-2-IM-support} to $f$. 
It then follows that $f$ must have one of the following forms: $(0,b,c,0)$, $(0,b,c,d)$, and $(a,b,c,0)$. Since $f\not\in\ED$, $f$ cannot be of the form $(0,b,c,0)$. In all the other cases, Lemma \ref{0-b-c-case} establishes 
an AP-reduction from $\sharpcspstar(OR,\FF)$ to $\sharpcspstar(f,\FF)$ for any constraint set $\FF$.   

Finally, assume that $k\geq3$. Now, we want to build   
a constraint $g\not\in\ED\cup\NZ$ of arity two such that 
$g\leq_{con} f$ and $g$ has no imp support.
Since $R_f\not\in IMP\cup\NZ$,  
Lemma  \ref{DGJ10-IMP-property} supplies two vectors  
$a=(a_1,\ldots,a_k)$ and $b=(b_1,\ldots,b_k)$ in $R_f$ satisfying 
either $a\wedge b\not\in R_f$ or $a\vee b\not\in R_f$ (or both). 
First, we will claim that (*) there are indices $i,j\in[k]$ such that $(a_i,b_i)=(0,1)$ and $(a_j,b_j)=(1,0)$. Assume otherwise; namely, 
either $(a_i,b_i)\in\{(0,1),(0,0),(1,1)\}$ for all $i\in[k]$ or $(a_i,b_i)\in\{(1,0),(0,0),(1,1)\}$ for all $i\in[k]$.  Let us consider the case where $a\vee b\not\in R_f$.  It easily  
follows that either $a = a\vee b$ or $b=a\vee b$ holds. This is a contradiction against the assumption $a\vee b\not\in R_f$. The other case where $a\wedge b\not\in R_f$ is similarly treated. Therefore, the claim (*) should hold.  

Hereafter, we assume that $a\vee b\not\in R_f$ since the other case 
(\ie $a\wedge b\not\in R_f$) is similarly handled. 
For simplicity, let $(a_1,b_1)=(0,1)$ and $(a_2,b_2)=(1,0)$. Now, we recursively define a new constraint $g$. Initially, we set $f_2= f$. 
If $f_{i-1}$ ($3\leq i\leq k$) has been already defined, 
then we define $f_{i}$ as follows. 
For each bit $c\in\{0,1\}$, if $(a_i,b_i)=(c,c)$, then  set  
$f_{i}= f_{i-1}^{x_i=c}$. If $(a_i,b_i)=(a_1,b_1)$, then let 
$f_{i} = f_{i-1}^{x_i=x_1}$. If $(a_i,b_i)=(a_2,b_2)$, then let 
$f_{i} = f_{i-1}^{x_i=x_2}$. Finally, we define $g$ to be  $f_{k}$. 
By this construction of $g$, $(0,1)$ and $(1,0)$ are 
in $R_g$; however, $(1,1)$ 
is not in $R_f$ because $a\vee b= (1,1,c_3,c_4,\ldots,c_k)\not\in R_f$ (for certain bits $c_i$'s) implies $g(1,1)=0$. In summary, it holds that $g(0,1) g(1,0)\neq0$ and 
$g(0,0) g(1,1)=0$. Lemma \ref{arity-2-IM-support} then concludes  
that $g$ is not in $\IM\cup\NZ$.  In particular, since $g$ is of arity two, $g$ has no imp support by Lemma \ref{arity-2-IM-support}. 
Moreover, the above construction is actually {\em T-construction}, and thus this fact ensures that $g\leq_{con} f$. Because this T-construction obviously uses no projection operation, by Lemma \ref{IMP-substitute}(2), $f\not\in\ED$ implies $g\not\in\ED$.  
To end our proof, we will claim that $g(0,0)\neq0$. Assume otherwise; namely,  $g$ has the form $(0,x,y,0)$ with $xy\neq0$. Obviously, $g$ belongs to $\ED$, a contradiction. Hence, $g(0,0)\neq0$ holds. We then conclude that  $g$ equals $(w,x,y,0)$ for certain non-zero constants $x,y,w$. By Lemma \ref{0-b-c-case}, it follows that $\sharpcspstar(OR,\FF)\APreduces \sharpcspstar(g,\FF)$. Since $g\leq_{con}f$, we obtain the desired consequence.    
\end{proofof}

\section{Dichotomy Theorem}\label{sec:main-theorem}

Our dichotomy theorem states that all counting problems of the form $\sharpcspstar(\FF)$ can be classified into exactly two categories, one of which consists of polynomial-time solvable problems and the other consists of $\#\mathrm{SAT}_{\complex}$-hard problems, assuming that $\#\mathrm{SAT}_{\complex}\not\in\fp_{\complex}$.   
This theorem steps forward in a direction toward a full analysis of a more general form of constraints than Boolean constraints.  The theorem also gives an approximation version of the dichotomy theorem of Cai \etal for exact counting problems (described in Section \ref{sec:introduction}).  Here, we rephrase our main theorem given in Section \ref{sec:introduction} as follows. 

\ms
\n{\bf {\em Theorem 1.1}} (rephrased)\hs{2}
{\em Let $\FF$ be any set of constraints. If 
$\FF\subseteq \ED$, then $\sharpcspstar(\FF)$ is in $\fp_{\complex}$. 
Otherwise, $\#\mathrm{SAT}_{\complex} \APreduces \sharpcspstar(\FF)$ holds.}
\ms

Through Sections \ref{sec:constraint-set} to \ref{sec:IM2-support}, we have developed necessary foundations to the proof of this dichotomy theorem, and now we are ready to apply them properly to prove the theorem. A center point of the proof of the theorem is the next proposition. To simplify a later discussion, the proposition targets only a single constraint, instead of a set of constraints as in the theorem.  


\begin{proposition}\label{key-proposition}
Let $f$ be any constraint. 
If $f$ is not in $\ED$, then  $\sharpcspstar(OR,\FF)\APreduces \sharpcspstar(f,\FF)$ holds for any set $\FF$ of constraints.
\end{proposition}

\begin{proof}
Let $f$ be any constraint not in $\ED$. Moreover, let $\FF$ be any constraint set.  We want to establish an AP-reduction from $\sharpcspstar(OR,\FF)$ to $\sharpcspstar(f,\FF)$.  
First, suppose that $f$ has imp support.  
Since $f\not\in\ED$, we apply Proposition \ref{no-affine-and-IM2} and instantly obtain the desired AP-reduction from  
$\sharpcspstar(OR,\FF)$ to $\sharpcspstar(f,\FF)$, as requested. 
Next, suppose that $f$ has no imp support.  
To finish the proof, we hereafter consider two independent cases.  

[Case: $f\not\in\NZ$] Since $f\not\in\ED\cup\NZ$, Proposition \ref{IM-XOR-IMP} leads to an AP-reduction from $\sharpcspstar(OR,\FF)$ to 
$\sharpcspstar(f,\FF)$. 

[Case: $f\in \NZ$] 
Notice that $f\not\in\DG$ since $\DG\subseteq\ED$.   
Lemma \ref{affine-PP-induction} provides a constraint 
$h=(1,x,y,z)$ satisfying that $xyz\neq0$, $z\neq xy$, and $h\leq_{con}f$. 
To this $h$, we apply Proposition \ref{1-x-y-z-case}, from which it follows that 
$\sharpcspstar(OR,\FF)\APreduces \sharpcspstar(h,\FF)$. Since $h\leq_{con}f$, Lemma \ref{T-con-to-AP-reduction} implies $\sharpcspstar(h,\FF)\APreduces \sharpcsp(f,\FF)$. Combining those two AP-reductions, we obtain the desired AP-reduction from $\sharpcspstar(OR,\FF)$ to $\sharpcspstar(f,\FF)$. 

Therefore, we have completed the proof.
\end{proof}

Finally, we give the long-awaited proof of Theorem \ref{dichotomy-theorem} and accomplish the main task of this paper.

\begin{proofof}{Theorem \ref{dichotomy-theorem}}
Let $\FF$ be any constraint set. If $\FF\subseteq\ED$, then 
Lemma \ref{basic-case-FPC} implies that $\sharpcspstar(\FF)$ 
belongs to $\fp_{\complex}$. 
Henceforth, we assume that  $\FF\nsubseteq \ED$. 
{}From this assumption, we choose a constraint $f\in\FF$ for which $f\not\in\ED$.  
Proposition \ref{key-proposition} then yields the AP-reduction:  $\sharpcspstar(OR)\APreduces \sharpcspstar(f)$.  
Since $f\in\FF$, it holds that 
$\sharpcspstar(f)\APreduces \sharpcspstar(\FF)$. By the transitivity of AP-reducibility,   
$\sharpcspstar(OR)\APreduces \sharpcspstar(\FF)$ follows. Note that, 
by Lemma \ref{SAT-to-OR}, we obtain  $\#\mathrm{SAT}_{\complex} \APreduces \sharpcspstar(OR)$.   Therefore, we conclude that  $\#\mathrm{SAT}_{\complex}$ is AP-reducible to $\sharpcspstar(\FF)$. 
\end{proofof}


As demonstrated in Theorem \ref{dichotomy-theorem}, a free use of unary constraint helps us obtain a truly stronger claim---dichotomy theorem---than a trichotomy theorem of Dyer \etalc~\cite{DGJ09} on unweighted 
Boolean \#CSPs. Is this phenomenon an indication that we could eventually prove a similar type of {\em dichotomy theorem} for all weighted Boolean \#CSPs? 
In our dichotomy theorem, we have shown that all seemingly intractable   \#CSPs are at least as hard as $\#\mathrm{SAT}_{\complex}$. Are those  problems are all AP-equivalent to $\#\mathrm{SAT}_{\complex}$?  Those questions demonstrate that we still have a long way to acquire a full understanding of the approximation complexity of the weighted \#CSPs. 


Next, we wish to prove Corollary \ref{algebraic-main-theorem}, which strengthens Theorem \ref{dichotomy-theorem} when limiting the free use of ``arbitrary'' constraints within ``algebraic'' constraints. Recall from Section \ref{sec:introduction} that an algebraic constraint outputs only algebraic complex numbers. Moreover, we  recall the notations  $\sharpcspplus(\FF)$ and $\#\mathrm{SAT}_{\algebraic}$ from Section \ref{sec:upper-bound}. Similarly, we write $\sharpcspstar_{\algebraic}(\FF)$ to denote $\sharpcspstar(\FF)$ whose instances are only algebraic constraints. 
Now, let us re-state the corollary given in Section \ref{sec:introduction}. 

\ms
\n{\bf {\em Corollary 1.2}} (rephrased)\hs{2}
{\em Let $\FF$ be any set of constraints. If $\FF\subseteq\ED$, then $\sharpcspplus_{\algebraic}(\FF)$ is in $\fp_{\algebraic}$; otherwise, $\#\mathrm{SAT}_{\algebraic}\APreduces \sharpcspplus_{\algebraic}(\FF)$ holds. }
\ms

Earlier, Dyer \etalc~\cite{DGJ09} demonstrated how to eliminate two constant constraints---$\Delta_0$ and $\Delta_1$---using randomized 
approximation schemes for unweighted Boolean \#CSPs. Similarly,   
we can eliminate those two constraints and thus prove the corollary by approximating them by 
any two non-zero constraints of the following forms: $[1,\lambda]$ and $[\lambda,1]$ with $|\lambda|<1$. 
In Lemma \ref{equivalent-plus-star}, we will demonstrate how to eliminate from $\sharpcspplus_{\algebraic}(\FF)$ all unary constraints whose output values contain zeros. This elimination is possible by a use of AP-reductions and this exemplifies a significance of the AP-reducibility.   

\begin{lemma}\label{equivalent-plus-star}
For any constraint set $\FF$, it holds that $\sharpcspstar_{\algebraic}(\FF) \APequiv \sharpcspplus_{\algebraic}(\FF)$. 
\end{lemma}

An argument of Dyer \etalc~\cite{DGGJ03} for their claim of eliminating both  $\Delta_0$ and $\Delta_1$ exploits their use of non-negative integers. However, since our target is {\em arbitrary} (algebraic) complex numbers, the proof of Lemma \ref{equivalent-plus-star} demands a quite different argument.  To make the paper readable, we postpone the proof until the last section.  Finally, we will give the proof of Corollary \ref{algebraic-main-theorem} using Lemma \ref{equivalent-plus-star}. 

\begin{proofof}{Corollary \ref{algebraic-main-theorem}}
Using Lemma \ref{equivalent-plus-star}, all results in this paper on $\sharpcspstar(\FF)$'s can be restated in terms of  $\sharpcspplus_{\algebraic}(\FF)$'s. Therefore, we obtain from Theorem \ref{dichotomy-theorem} that $\sharpcspplus_{\algebraic}(\FF)\in\fp_{\algebraic}$ if $\FF\subseteq\ED$, and $\sharpcspplus_{\algebraic}(OR)\APreduces \sharpcspplus_{\algebraic}(\FF)$ otherwise. It thus remains to show that $\#\mathrm{SAT}_{\algebraic} \APreduces \sharpcspplus_{\algebraic}(OR)$. 

As remarked in the end of Section \ref{sec:upper-bound}, following a similar argument given in the proof of Lemma \ref{SAT-to-OR}, it is possible to prove that $\#\mathrm{SAT}_{\algebraic}$ is AP-reducible to $\sharpcspstar_{\algebraic}(OR)$.  By Lemma \ref{equivalent-plus-star}, it follows that  $\#\mathrm{SAT}_{\algebraic} \APreduces \sharpcspplus_{\algebraic}(OR)$.  Therefore, the corollary holds. 
\end{proofof}

\section{Proofs of Lemmas \ref{constructibility} and \ref{equivalent-plus-star}}\label{sec:elimination}

This last section will fill the missing proofs of Sections \ref{sec:constructibility} and \ref{sec:main-theorem} to complete the proofs of our main theorem and its corollary. 
First, we will give the proof of Lemma \ref{equivalent-plus-star}. 
A use of algebraic numbers in the lemma ensures the correctness of a randomized approximation scheme  
used in the proof of the lemma. Underlying ideas of the scheme   
come from the proofs of \cite[Lemma 10]{DGJ09} and \cite[Theorem 3(2)]{Yam03}. Particularly, the latter relied on the following well-known 
lower bound of the absolute values of polynomials in algebraic numbers. 

\begin{lemma}\label{complex-lower-bound}{\rm \cite{Sto74}}\hs{1}
Let $\alpha_1,\ldots,\alpha_m\in\algebraic$ and let $c$ be the degree of $\rational(\alpha_1,\ldots,\alpha_m)/\rational$. There exists a constant $e>0$ that satisfies the following statement for any complex number $\alpha$ of the form $\sum_{k}a_{k}\left(\prod_{i=1}^{m}\alpha_i^{k_i}\right)$, where $k=(k_1,\ldots,k_m)$ ranges over $[N_1]\times\cdots\times[N_m]$, $(N_1,\ldots,N_m)\in\nat^{m}$, and $a_k\in\integer$. If $\alpha\neq0$, then $|\alpha|\geq \left(\sum_{k}|a_k|\right)^{1-c}\prod_{i=1}^{m}e^{-cN_i}$. 
\end{lemma}

Now, we start the proof of Lemma \ref{equivalent-plus-star}.

\begin{proofof}{Lemma \ref{equivalent-plus-star}}
Since any input instance to $\sharpcspplus_{\algebraic}(\FF)$ is obviously an instance to $\sharpcspstar_{\algebraic}(\FF)$, it easily follows that $\sharpcspplus_{\algebraic}(\FF)$ is AP-reduced to $\sharpcspstar_{\algebraic}(\FF)$. Hereafter, we wish to prove the other direction, namely,  $\sharpcspstar_{\algebraic}(\FF)\APreduces 
\sharpcspplus_{\algebraic}(\FF)$.  Since    $\sharpcspstar_{\algebraic}(\FF)$ coincides with  $\sharpcspplus_{\algebraic}(\FF,\Delta_0,\Delta_1)$, we wish to 
demonstrate  how to eliminate $\Delta_0$ from $\sharpcspplus_{\algebraic}(\FF,\Delta_0,\Delta_1)$.   
Without loss of generality, we aim at proving   
that $\sharpcspplus_{\algebraic}(\FF,\Delta_0)$ is AP-reducible to $\sharpcspplus_{\algebraic}(\FF)$. 

Let $\Omega=(G,X|\FF',\pi)$ be any constraint frame given as an input instance to  $\sharpcspplus_{\algebraic}(\FF,\Delta_0)$, where $G=(V_1|V_2,E)$, $X=\{x_1,x_2,\ldots,x_n\}$, and $\FF'\subseteq \FF\cup\{\Delta_0\}$. Let $n$ be the number of distinct variables used in $G$. If $\FF$ contains $\Delta_0$, the lemma is trivially true. Henceforth, we assume that $\Delta_0\not\in \FF$.
Choose any complex number $\lambda$ satisfying $0<|\lambda|<1$ and 
define $u(x)=[1,\lambda]$, which is clearly in $\UU\cap\NZ$. 
For later use, let $|\Omega|$ denote $\prod_{v\in V_2}\max\{1,|f_v|\}$, where $|f_v|= \max\{|f_v(x)| \mid x\in\{0,1\}^k\}$ and $k$ is the arity of $f_v$. 

First, we will modify the graph $G$ as follows. Let us select all nodes in $V_1$ that are adjacent to certain nodes in $V_2$ having the label $\Delta_0$. 
We first merge all selected variable nodes into a single node, say, 
$v$ with a ``new'' label, say, $x$, and then delete all the nodes labeled $\Delta_0$ and their incident edges. Finally, we attach a ``new'' node labeled $\Delta_0$ to  the 
node $v$ by an additional single edge. It is not difficult to show that this modified graph produces the same output value as its original one. In what follows, we assume that the constraint $\Delta_0$ appears exactly once as a node label in the graph $G$ and it depends only on the variable $x$. 

Let $G_0$ be the graph obtained from $G$ by removing the unique node $\Delta_0$. Its associated constraint frame is briefly denoted $\Omega_0$. Note that $\csp_{\Omega_0}$ can be expressed as $\sum_{x\in\{0,1\}}h(x)$  using an appropriate complex-valued function $h$ depending on the value of 
$x$. With this $h$, $\csp_{\Omega}$ is calculated as   $\sum_{x\in\{0,1\}}h(x)\Delta_0(x)$, which obviously equals $h(0)$. 
Moreover, let $u_m = u^m$ for any fixed number $m\in\nat^{+}$. Denote by $G_m$ the graph obtained from $G$ by replacing $\Delta_0$ by $u_{m}$ and let $\Omega_m$ be its associated constraint frame.  Since $u_m=[1,\lambda^m]$, 
it holds that  
$\csp_{\Omega_m}=\sum_{x}h(x)u_m(x) = h(0)+\lambda^m h(1)$.  
Letting $K= h(1)$, we obtain  $\csp_{\Omega} = \csp_{\Omega_m} - \lambda^{m}K$. Note that, for each fixed variable assignment $\sigma:X\rightarrow\{0,1\}$, the product of the outcomes of all constraints is at most $|\Omega|$. Since  there are $2^n$ distinct variable assignments, $|K|$ is thus upper-bounded by 
$2^n|\Omega|$.  

Meanwhile, we assume that $\csp_{\Omega}\neq0$. Since all entries of any  constraint in $\FF'$ are taken from $\algebraic$, we want to apply Lemma \ref{complex-lower-bound}. For a use of this lemma, however, we need to express the value $|\csp_{\Omega}|$ using three series $\{a_{k}\}_{k}$, $\{\alpha_{i}\}_{i}$, and $\{k_i\}_{i}$ given in the lemma. Let us define them as follows.  
Let $I=\{\pair{v,w}\mid v\in V_2, w\in\{0,1\}^{r}\}$, where $r$ is the arity of $f_v$. Here, we assume a fixed enumeration of all elements in $I$. For each variable assignment $\sigma:X\rightarrow\{0,1\}$, we define  a vector  $k^{(\sigma)}=(k^{(\sigma)}_{i})_{i\in I}$ as follows: for each $i=\pair{v,w}\in I$, let $k^{(\sigma)}_{i} =1$ if $f_v$ depends on a certain variable series $(x_{i_1},\ldots,x_{i_r})$ and $w$ equals $(\sigma(x_{i_1}),\ldots,\sigma(x_{i_k}))$; otherwise, let $k^{(\sigma)}_{i}=0$. 
Moreover, let $N_i$ ($i\in I$) equal $1$ and set $N=\prod_{i\in I}[N_i]$.  For any vector $k\in N$, let $a_{k} = 1$ if there exists a valid assignment $\sigma$ satisfying $k=k^{(\sigma)}$; otherwise, let $a_{k}=0$. 
Finally, let $\alpha_{\pair{v,w}}= f_{v}(w)$, where $\pair{v,w}\in I$. 
By these definitions, the value $\csp_{\Omega} = \sum_{\sigma}\prod_{v\in  V_2}f_{v}(\sigma(x_{i_1}),\ldots,\sigma(x_{i_k}))$ equals $\sum_{k\in N}a_{k}\left(\prod_{i\in I}\alpha^{k_i}_{i}\right)$. 
Now, Lemma \ref{complex-lower-bound} provides two constants $c,e>0$ for which $|\csp_{\Omega}|$ is lower-bounded by the value $\left(\sum_{k\in N}a_k\right)^{1-c}\prod_{i\in I}e^{-cN_i}$. For our purpose, we set $d= (1/2)\left(\sum_{k\in N}a_k\right)^{1-c}\prod_{i\in I}e^{-cN_i}$, from which we obtain $|\csp_{\Omega}|>d$. 
For convenience, whenever $d\geq1$, we automatically set $d=1/2$ so that we can always assume that $0<d<1$. 

Now, we will describe an AP-reduction $N$ from $\sharpcspplus_{\algebraic}(\FF,\Delta_0)$ to $\sharpcspplus_{\algebraic}(\FF)$.  Let $(\Omega,1/\epsilon)$ be any input to the algorithm $N$, where $\epsilon\in(0,1)$ is an error tolerance parameter. Without loss of generality, we assume that $0<\epsilon<1/4$.   
Let $M$ be an arbitrary oracle that is a randomized approximation scheme  designed to solve $\sharpcspplus_{\algebraic}(\FF)$. 
To compute the value $\csp_{\Omega}$ on the given input $(\Omega,1/\epsilon)$,  $N$ works as follows.  

\begin{quote}
Let $\delta =\epsilon/2$ and choose a positive integer $m$ for which (i) $2^{n+\delta}|\Omega||\lambda|^m \leq (1-2^{-\delta})d$ and (ii) $2^n|\Omega||\lambda|^m \leq \delta' d$, where $n =|X|$ and $\delta' = \frac{\pi}{3}\delta$.  
Next, $N$ constructs the constraint frame $\Omega_m$ from $\Omega$.
To the oracle $M$, $N$ makes a query with a query word $(\Omega_m,1/\delta)$. 
Let $z$ denote an oracle answer from $M$. Notice that $z$ is a random variable.  If $|z|<d$, then $N$ outputs $0$; otherwise, it outputs $z$. 
\end{quote}

The correctness of the above algorithm $N$ is shown as follows. 
Let us focus on the case where $\csp_{\Omega}=0$; in other words, 
$\csp_{\Omega_m} - \lambda^mK =0$, which is equivalent to $\csp_{\Omega_m}= \lambda^mK$. 
Since $z$ is a $2^{\delta}$-approximate solution for $\csp_{\Omega_m}$, 
the value $z$ must satisfy that  $2^{-\delta}|\csp_{\Omega_m}|\leq |z|\leq 2^{\delta}|\csp_{\Omega_m}|$. 
It thus follows by Condition (i) that    
\[
|z| \leq 2^{\delta}|\csp_{\Omega_m}| \leq 2^{\delta}|\lambda^mK|
 \leq 2^{n+\delta}|\Omega||\lambda^m| \leq (1-2^{-\delta})d < d.
\] 
In this case, the algorithm $N$ outputs $0$, that is, $N(\Omega,1/\varepsilon)=0$. This means that $N$ correctly computes $\csp_{\Omega}$ with high probability.  

Let us consider the other case where $\csp_{\Omega}\neq0$.  
Due to the choice of $d$, $|\csp_{\Omega}|>d$ holds.  
Since the oracle returns a $2^{\delta}$-approximate solution $z$ for 
$\csp_{\Omega_m}$, it follows that $2^{-\delta}|\csp_{\Omega_m}|\leq |z| \leq 2^{\delta}|\csp_{\Omega_m}|$. 
Now, we want to show that (iii) $2^{-\epsilon}|\csp_{\Omega}|\leq 2^{-\delta}(|\csp_{\Omega}|-|\lambda^mK|)$ and 
(iv) $2^{\delta}(|\csp_{\Omega}|+|\lambda^mK|)\leq 2^{\epsilon}|\csp_{\Omega}|$, because these bounds together imply 
\[
2^{-\epsilon}|\csp_{\Omega}|\leq 2^{-\delta}(|\csp_{\Omega}|-|\lambda^mK|) \leq  |z|  
\leq  2^{\delta}(|\csp_{\Omega}|+|\lambda^mK|)\leq 2^{\epsilon}|\csp_{\Omega}|.
\]
In other words, $2^{-\epsilon}\leq \left| z/\csp_{\Omega} \right|\leq 2^{\epsilon}$ holds. 
Our next task is to prove  Conditions (iii) and (iv). 
Condition (iii) is equivalent to $|\lambda^mK|\leq (1-2^{-\delta})|\csp_{\Omega}|$, whereas Condition (iv) is equivalent to $|\lambda^mK| \leq (2^{\delta}-1)|\csp_{\Omega}|$.  
Since $2^{\delta}-1\geq 1-2^{-\delta}$ holds for our choice of $\delta$, Condition (iv) follows instantly from Condition (iii).   
By Condition (i), we obtain  $|\lambda^mK|\leq 2^n|\Omega||\lambda|^m\leq (1-2^{-\delta})d$.  This immediately implies Condition (iii) since $|\csp_{\Omega}|>d$. 

To complete the proof, we still need to show that $|\arg(z/\csp_{\Omega})|\leq\varepsilon$ whenever $\csp_{\Omega}\neq0$.   Let us assume that $\csp_{\Omega}\neq0$.  Since $z$ is an output of $M$ on the input $(\Omega_{m},1/\delta)$, it holds that $|\arg(z)-\arg(\csp_{\Omega_m})|\leq \delta$. Now, we set $\theta=|\arg(\csp_{\Omega_m}) - \arg(\csp_{\Omega})|$. Notice that the value $\theta$ represents an angle in the complex plane between two ``vectors'' $\csp_{\Omega_m}$ and $\csp_{\Omega}$.  Since $\csp_{\Omega_m} = \csp_{\Omega}+\lambda^mK$, the value  $\theta$ is maximized when the vector  $\lambda^mK$ is perpendicular to the line extending the vector  $\csp_{\Omega_m}$. This implies that $|\csp_{\Omega}|\sin\theta \leq |\lambda^mK|$.   
Condition (ii) implies that $|\lambda^mK|\leq 2^n|\Omega||\lambda|^m\leq \delta'd$. Because $|\csp_{\Omega}|>d$, we also obtain $\sin\theta\leq \frac{|\lambda^mK|}{|\csp_{\Omega}|}\leq \frac{\delta' d}{d} = \delta'$.  Since $\delta'=\frac{\pi}{3}\delta = \frac{\pi}{6}\varepsilon<\frac{\pi}{24}<\frac{1}{2}$, we may assume that $0\leq \theta\leq \frac{\pi}{6}$. Within this range, it always holds that $\frac{\pi}{3}\theta \leq \sin\theta$. Therefore, we conclude that $\theta \leq \frac{3}{\pi}\delta' = \delta$. By the triangle inequality, it thus follows that 
\[
|\arg(z)-\arg(\csp_{\Omega})| \leq |\arg(z)-\arg(\csp_{\Omega_m})| + |\arg(\csp_{\Omega_m})-\arg(\csp_{\Omega})|\leq \delta+\delta = \varepsilon.
\]

By following the above argument closely, it is also possible to prove that $\sharpcspplus_{\algebraic}(\FF,\Delta_1)$ is AP-reducible to $\sharpcspplus_{\algebraic}(\FF)$. Therefore, we have completed the proof of the lemma. 
\end{proofof}

In our argument toward the dichotomy theorem, we have omitted the proof of Lemma \ref{constructibility}, which shows a fundamental property of T-constructibility. Now, we will give the proof of the lemma.  The proof 
will proceed by induction on the number of operations applied to construct a target constraint. 

\begin{proofof}{Lemma \ref{constructibility}}
Let $\FF$ be any set of constraints. For simplicity, assume that $f$ is obtained from $g$ (or $\{g_1,g_2\}$ in the case of the multiplication operation) 
by applying exactly one of the seven operations described in Section \ref{sec:constructibility}.  Our purpose is to show that 
$\sharpcspstar(f,\FF)$ is AP-reducible to $\sharpcspstar(g,\FF)$. 
Notationally, we set $\Omega=(G,X|\HH,\pi)$ 
and $\Omega'=(G',X'|\HH',\pi')$ to be any  
constraint frames associated with $g$ and $f$, respectively. Note that  
$\HH$ and $\HH'$ are finite subsets of  $\{g\}\cup\FF\cup\UU$ and 
$\{f\}\cup\FF\cup\UU$, and $X$ and $X'$ are both finite sets of Boolean variables. To improve the readability, we assume that $X=X'=\{x_1,x_2,\ldots,x_n\}$.  Let $\epsilon$ be any error tolerance parameter. 
For each operation, we want to explain how to produce $G$ and $\pi$ from $G'$ and $\pi'$ in polynomial time so that, after making a query $(\Omega,1/\epsilon)$ to any oracle (which is a randomized approximation scheme solving $\sharpcspstar(g,\FF)$), from its oracle answer $\csp_{\Omega}$, we can compute a $2^{\epsilon}$-approximate solution for  $\csp_{\Omega'}$ with high probability. This procedure indicates that $\sharpcspstar(f,\FF)\APreduces \sharpcspstar(g,\FF)$. For simplicity, we  will omit the mentioning of $\epsilon$ in the following description. 

\s

[{\sc Permutation}] Assume that $f$ is obtained from $g$ by exchanging two  indices $i$ and $j$ of variables $\{x_i,x_j\}$. 
{}From $G'$, we build $G$ by swapping only the labels $x_i$ and $x_j$ of the corresponding nodes (without changing any edge incident on them). The labeling function $\pi$ is also naturally induced from $\pi'$. 
Clearly, this step requires linear time.  
Our underlying randomized approximation scheme $N$ works as follows: it first constructs $\Omega$ from $\Omega'$, makes a single query to obtain a  $2^{\epsilon}$-approximate solution $z$ for $\csp_{\Omega}$ from the oracle $\sharpcspstar(g,\FF)$, and outputs $z$ instantly. 
Since $\csp_{\Omega'} = \csp_{\Omega}$, the output of $N$ is also a $2^{\epsilon}$-approximate solution for $\csp_{\Omega'}$. 

[{\sc Pinning}] Let $f=g^{x_i=c}$ for $i\in[k]$ and $c\in\{0,1\}$. {}From $G'$, we construct $G$ in polynomial time as follows: append a new node whose label is $\Delta_c$ and connect it to  the node labeled $x_i$ by a new edge. Because $\csp_{\Omega'} = \csp_{\Omega}$ holds, an  algorithm similar to the one in the previous case can approximate $\csp_{\Omega'}$.  

[{\sc Projection}] Let $f = g^{x_i=*}$ with index $i\in[k]$. 
Notice that $f$ does not have the variable $x_i$. For simplicity, assume that $G'$ has no node with the label $x_i$. 
Let $V'$ denote the set of all nodes having the label $f$ in $G'$. 
Now, we construct $G$ from $G'$ by adding a new node labeled $x_i$ to $V_2$, 
by replacing the label $f$ by $g$, and by connecting the node $x_i$ to all nodes in $V'$ by new single edges. This transformation implies that $\csp_{\Omega'} = \csp_{\Omega}$.  The rest is similar to the previous cases. 

[{\sc Linking}] Let $f=g^{x_i=x_j}$ and assume that $i<j$. In this case, we obtain $G$ from $G'$  by replacing the label $f$ by $g$ and adding an extra  edge between the node $x_j$ and the node $g$. Note that there are 
now two different edges between the node $x_j$ and the node $g$. 
Using this new graph $G$, we obtain $\csp_{\Omega'} = \csp_{\Omega}$.  

[{\sc Expansion}] Let $f(x_1,\ldots,x_{i},y,x_{i+1},\ldots,x_k)= g(x_1,\ldots,x_i,x_{i+1},\ldots,x_k)$,  where $y$ is a free variable. 
To define $G$, we delete from $G'$ any edge between the node $y$ and any node labeled $f$. Note that, in general, we cannot remove the node $y$ from $G'$ because it may be connected to other nodes although $g$ does not have the variable $y$. Since any node labeled $f$ has not initially {\em depended} on the node $y$, 
it follows that $\csp_{\Omega'} = \csp_{\Omega}$. 
Since a $2^{\epsilon}$-approximate solution $z$ for $\csp_{\Omega}$ can be obtained as an answer to an oracle query regarding $\Omega$, $z$  approximates  $\csp_{\Omega'}$ as well.  

[{\sc Multiplication}] Assume that $f=g_1\cdot g_2$, 
where $g_1$ and $g_2$ share the same input variable series. 
We intend to define $G$ in the following way. 
Since $\{g_1,g_2\}$ is a factor list for $f$, we first replace every node labeled $f$ in $G'$ by two new nodes having the labels 
$g_1$ and $g_2$, 
each of which has the same incident set as $f$ does. Using  
$\sharpcspstar(g_1,g_2,\FF)$ as an oracle, we obtain a $2^{\epsilon}$-approximate solution $z$ for $\csp_{\Omega}$ as an oracle answer.  
Since $\csp_{\Omega'} = \csp_{\Omega}$, $z$ approximates $\csp_{\Omega'}$.  

[{\sc Normalization}] Let $f= \lambda \cdot g$ for a constant $\lambda\in\complex-\{0\}$. 
Define $G$ to be $G'$ except that every occurrence of $f$ is 
replaced by $g$.  Let $n$ be the number of nodes in $G'$ 
that have the label $f$. Since $\csp_{\Omega'} = \lambda ^n\cdot \csp_{\Omega}$,  we first obtain a $2^{\varepsilon}$-approximate solution $z$ for $\csp_{\Omega}$ by making a query to the oracle $\sharpcspstar(g,\FF)$. We then multiply $z$ by $\lambda^m$ and output the resulting value. Clearly, this value approximates $\csp_{\Omega'}$. 

\s

{}From all seven cases discussed above, we conclude that 
$\sharpcspstar(f,\FF)$ is AP-reducible to $\sharpcspstar(g,\FF)$. 
This finishes the proof of Lemma \ref{constructibility}.  
\end{proofof}

\paragraph{Acknowledgments:} 
The author is grateful to Leslie Goldberg for helping him understand the notion of AP-reductions. The author is also appreciative of Jin-Yi Cai's introducing him a theory of signature while he was visiting at the University of Wisconsin in February 2010. 

\bibliographystyle{alpha}

\end{document}